\documentclass[12pt]{article}

\usepackage{amsmath,amsthm,amssymb}
\usepackage{mathrsfs}
\usepackage{mathtools}
\usepackage{enumitem}
\usepackage{booktabs}
\usepackage{float}
\usepackage{url}
\usepackage{hyperref}
\usepackage{natbib}
\newtheorem{theorem}{Theorem}[section]
\newtheorem{lemma}[theorem]{Lemma}
\newtheorem{proposition}[theorem]{Proposition}
\newtheorem{corollary}[theorem]{Corollary}

\theoremstyle{definition}
\newtheorem{definition}[theorem]{Definition}
\newtheorem{example}[theorem]{Example}

\theoremstyle{remark}
\newtheorem{remark}[theorem]{Remark}

\newcommand{\bits}{\{0,1\}}
\newcommand{\F}{\mathbb{F}}
\newcommand{\N}{\mathbb{N}}
\newcommand{\poly}{\mathrm{poly}}
\newcommand{\codim}{\mathrm{codim}}
\newcommand{\SPDPspace}{\mathcal{V}}

\newcommand{\rankSPDP}{\Gamma}
\newcommand{\codimSPDP}{\mathrm{codim}}
\DeclareMathOperator{\rk}{rank}
\DeclareMathOperator{\rank}{rank}
\DeclareMathOperator{\Perm}{Perm}

\title{Shifted Partial Derivative Polynomial Rank and Codimension}
\author{Darren J.\ Edwards\\
\small Swansea University, United Kingdom\\
\small \texttt{d.j.edwards@swansea.ac.uk}}
\date{}

\begin{document}

\maketitle

\begin{abstract}
Shifted partial derivative (SPD) methods are a central algebraic tool for proving
circuit lower bounds by measuring the dimension of spaces of shifted derivatives
of a polynomial. In this paper we develop the Shifted Partial Derivative Polynomial
(SPDP) framework, packaging classical SPD methods into an explicit
coefficient-matrix formalism. This turns these spaces into concrete
linear-algebraic objects and introduces two dual complexity measures:
SPDP rank and SPDP codimension.

We first define the SPDP generating family, the associated span, and the SPDP
matrix $M_{\kappa,\ell}(p)$ in a fixed ambient coefficient space determined by the
$(\kappa,\ell)$ regime, so that rank is well-defined and codimension becomes a
canonical ``deficit from ambient fullness.'' We then prove structural properties
that make the framework usable as a standalone toolkit: monotonicity in the
shift/derivative parameters (with precise scoping for the $|S|=\kappa$ versus
$|S|\le \kappa$ conventions), invariance under admissible variable symmetries and
basis changes, and robustness across standard Boolean/multilinear embedding
conventions.

Next, we give generic width-to-rank upper-bound templates for local circuit
models, showing how combinatorial profile counting yields polynomial bounds
in logarithmic $(\kappa,\ell)$ regimes under the deterministic compiler/diagonal-basis
hypothesis. This explicitly separates the model-agnostic SPDP toolkit from
additional compiled refinements used in separate work (e.g. Edwards 2025);
nothing in this paper relies on those refinements. Finally, we illustrate the
codimension viewpoint on representative examples, emphasizing how codimension
captures rigidity as the failure of shifted-derivative spans to fill their
ambient space.

Overall, SPDP provides a clean, self-contained linear-algebraic packaging of
shifted partial derivatives that supports both rank-based and deficit-based
(codimension) reasoning, and is intended as a reusable reference for algebraic
complexity applications.
\end{abstract}


\tableofcontents
\newpage

\section{Introduction}
Shifted partial derivatives have become a standard tool in algebraic and Boolean
complexity, particularly for proving lower bounds against restricted circuit
models~\citep{NW1997,Kayal2012,GKKS2016,KayalSaha2016,LST2024,SY2010}. Most of the literature uses
\emph{dimension} of a shifted partial derivative space as the primary complexity
statistic. The present work isolates and formalizes a closely related invariant
package---the \emph{Shifted Partial Derivative Polynomial (SPDP)} framework---with
two dual quantities: \emph{SPDP rank} (dimension) and \emph{SPDP codimension}
(ambient deficit). For an application of these ideas to complexity separations,
see \citet{Edwards2025}.

\paragraph{Main contributions.}
This note isolates a clean, self-contained ``SPDP toolkit'' that can be cited
independently of any separation application.

\begin{itemize}
  \item \textbf{SPDP matrix formalism + dual invariant.}
  We package the shifted partial derivative (SPD) dimension measure as an explicit
  coefficient matrix $M_{\kappa,\ell}(p)$ whose rank is $\Gamma_{\kappa,\ell}(p)$, and we
  introduce the dual quantity $\mathrm{codim}_{\kappa,\ell}(p)$ as the ambient deficit.
  This makes invariance and monotonicity statements transparent as linear-algebraic
  properties of $M_{\kappa,\ell}(p)$.

  \item \textbf{Robustness and invariances.}
  We prove basis-invariance and permutation-invariance of $\Gamma_{\kappa,\ell}(p)$ and
  $\mathrm{codim}_{\kappa,\ell}(p)$ under the intended ambient conventions, clarifying the
  scope of invariance under Boolean/multilinear embeddings.

  \item \textbf{A width$\Rightarrow$rank upper-bound schema with a sharp profile count.}
  We formalize a minimal local-width computation model and prove a polynomial
  Width$\Rightarrow$Rank bound in the regime $\kappa,\ell=\Theta(\log n)$ and
  $R=C(\log n)^c$.
  The key technical bridge is a \emph{profile compression} lemma showing that the number
  of realizable interface-anonymous profiles is $R^{O(1)}$ \emph{independent of the
  window length $\kappa$}, avoiding the crude ordered-step count $(\log n)^{O(\kappa)}$.

  \item \textbf{Compatibility with the block-partitioned SPDP framework.}
  We record a dictionary relating the unblocked SPDP matrix $M_{\kappa,\ell}(p)$ to the
  block-partitioned matrix $M^B_{\kappa,\ell}(p)$ used in block-local/compiled settings,
  including the rank comparison $\Gamma^B_{\kappa,\ell}(p)\le \Gamma_{\kappa,\ell}(p)$.
\end{itemize}

\paragraph{What this paper does}
We provide (i) a clean definition of SPDP spaces and a canonical SPDP matrix
representation; (ii) rank and codimension invariants that are stable under
natural transformations; (iii) general properties (monotonicity, invariance,
and basic closure lemmas); and (iv) rank upper bounds for broad computational
models in a consistent parameter regime, together with codimension formulations
that clarify rigidity.

\paragraph{What this paper does not do}
This paper is \emph{not} a separation claim and does not address
barrier frameworks or any global complexity-class consequences.
The aim is to establish SPDP rank and codimension as a standalone
measure with a rigorous, self-contained development.

\paragraph{Application to other work (compiled/block-partitioned SPDP).}
This note develops SPDP rank and codimension in the \emph{unblocked} matrix formalism,
i.e.\ via the coefficient matrix $M_{\kappa,\ell}(p)$ and its rank $\Gamma_{\kappa,\ell}(p)$ (and the dual
ambient deficit $\mathrm{codim}_{\kappa,\ell}(p)$) under a fixed ambient/basis convention.
Other work~\citep{Edwards2025} applies a more restrictive \emph{compiled} variant in which a fixed
radius--$1$ block partition $B$ is imposed and one takes the rank of a structured restriction
$M^B_{\kappa,\ell}(p)$ (with $\Gamma^B_{\kappa,\ell}(p)\le \Gamma_{\kappa,\ell}(p)$) tailored to block-local
compiler transformations. Nothing in the present paper relies on that compiled setting;
we include a short dictionary (\ref{subsec:compat-spdp-cew}) only to align notation.

\section{Preliminaries and Notation}
Fix a base field $\F$ (e.g.\ $\mathbb{Q}$ or $\mathbb{R}$). We work with
polynomials $f \in \F[x_1,\dots,x_n]$ representing Boolean functions on
$\bits^n$ (via multilinearization or a chosen embedding, made explicit in
\ref{sec:embedding}).

\subsection{Multi-index notation}
For $\alpha = (\alpha_1,\dots,\alpha_n)\in\N^n$, let $|\alpha|=\sum_i \alpha_i$
and write $\partial^\alpha = \partial_{x_1}^{\alpha_1}\cdots \partial_{x_n}^{\alpha_n}$.

\subsection{Parameter regime}
Throughout, we fix parameters $(\kappa,\ell)$ that may depend on $n$.
All theorems in this paper explicitly state which $(\kappa,\ell)$ regime they assume.
When helpful, we highlight standard regimes, e.g.\ $\kappa=O(1)$ with $\ell=\poly(n)$,
or $\kappa=\Theta(\log n)$ with $\ell=\Theta(\log n)$. No claim in this paper depends
on mixing incompatible parameter choices.

\section{Polynomial embedding of Boolean functions}
\label{sec:embedding}
We record a minimal embedding convention. The reader may substitute an equivalent
standard convention; our invariance results show the measure is robust to such
choices.

\begin{definition}[Boolean-to-polynomial representation]
\label{def:spdp-boolean}
Let $f:\bits^n\to\bits$. A \emph{polynomial representation} of $f$ is a polynomial
$\tilde f\in\F[x_1,\dots,x_n]$ that agrees with $f$ on $\bits^n$.
Unless otherwise stated, we assume $\tilde f$ is multilinear (via standard
multilinearization modulo $(x_i^2-x_i)$).
\end{definition}

\begin{remark}
\label{rem:non-multilinear}
If a non-multilinear representation is used, all SPDP spaces below may be taken
modulo the Boolean ideal $(x_i^2-x_i)$; this does not affect the rank/codimension
definitions, only the ambient basis choice.
\end{remark}

\section{The SPDP framework}
\label{sec:spdp-framework}

Shifted partial derivative methods study the linear span of low-order partial
derivatives of a target polynomial, optionally multiplied by low-degree monomials
(shift). In this paper we package that construction as an explicit matrix and
work with two dual invariants: SPDP rank and SPDP codimension.

\subsection{Parameter regime and scope}
\label{subsec:params-scope}

Fix input length $n$ and a base field $\F$. Let $p \in \F[x_1,\dots,x_n]$ be a
polynomial representing a Boolean function under a chosen embedding (e.g.\ multilinear
representation modulo $(x_i^2-x_i)$). We work with parameters $(\kappa,\ell)$, where
$\kappa$ is the derivative order and $\ell$ is the degree bound for the shift monomials.

A key methodological requirement is \emph{parameter consistency}: whenever we compare
an upper bound for a computational model to a lower bound for an explicit family, both
bounds must hold for the \emph{same} $(\kappa,\ell)$ regime.

\begin{remark}[Typical regimes]
Two standard choices are:
\begin{enumerate}[label=(\roman*),leftmargin=2.4em]
\item $\kappa,\ell = \Theta(\log n)$ (polylogarithmic derivatives/shifts), often used when
one expects polynomial bounds under suitable structural assumptions.
\item $\kappa = \Theta(n)$ with $\ell=O(1)$ (high-order derivatives, no/low shift), used in
classical partial-derivative rank arguments for explicit polynomials.
\end{enumerate}
This paper is agnostic: all statements below specify the regime they use.
\end{remark}

\subsection{Relation to classical shifted partial derivatives (SPD) and what ``SPDP'' adds}
\label{subsec:spdp-vs-spd}

Shifted partial derivatives (SPD) are typically presented as a \emph{space}
dimension measure: one takes low-order partial derivatives of a target
polynomial and then ``shifts'' them by multiplying with low-degree monomials,
and measures the dimension of the resulting linear span
\citep[e.g.][]{NW1997,SY2010}.
The SPDP framework used in this paper is \emph{not a different measure in disguise};
it is a matrix-based packaging of the same underlying SPD span, together with a
canonical ambient-basis convention that makes \emph{codimension} a first-class dual
invariant.

\paragraph{Classical SPD space}
Let $p \in \mathbb{F}[x_1,\dots,x_n]$ (or its multilinear representative under the
Boolean embedding of \ref{sec:embedding}). Fix parameters $(\kappa,\ell)$.
Define the classical shifted partial derivative space (at order $\kappa$ and shift degree $\ell$) as
\begin{equation}
\label{eq:spd-space}
\mathsf{SPD}_{\kappa,\ell}(p)
\;:=\;
\mathrm{span}_{\mathbb{F}}
\bigl\{\, m \cdot \partial^\alpha p \;:\; |\alpha|=\kappa,\; \deg(m)\le \ell \,\bigr\}.
\end{equation}
In the Boolean/multilinear setting it is common (and consistent with our conventions)
to restrict $m$ to multilinear monomials and to interpret all polynomials modulo the
Boolean ideal $\langle x_i^2-x_i\rangle$ (\ref{rem:non-multilinear}).

The classical SPD \emph{complexity statistic} is then the dimension
\[
\mathrm{SPDdim}_{\kappa,\ell}(p) \;:=\; \dim_{\mathbb{F}}\bigl(\mathsf{SPD}_{\kappa,\ell}(p)\bigr).
\]

\paragraph{SPDP is the same span, presented as a rank}
In \ref{sec:spdp-matrix-rank-codim} we define the SPDP generating family
$\mathcal{G}_{\kappa,\ell}(p)$ and its span $V_{\kappa,\ell}(p)$, and then package that span
into an explicit coefficient matrix $M_{\kappa,\ell}(p)$ in a fixed ambient monomial basis.
Concretely,
\begin{equation}
\label{eq:spdp-equals-spd-span}
V_{\kappa,\ell}(p)
\;=\;
\mathrm{span}_{\mathbb{F}}\bigl(\mathcal{G}_{\kappa,\ell}(p)\bigr)
\;=\;
\mathsf{SPD}_{\kappa,\ell}(p),
\end{equation}
and therefore the SPDP rank defined in \ref{sec:spdp-matrix-rank-codim} agrees
with the classical SPD dimension:
\begin{equation}
\label{eq:spdp-rank-equals-spddim}
\Gamma_{\kappa,\ell}(p)
\;=\;
\mathrm{rank}_{\mathbb{F}}\!\bigl(M_{\kappa,\ell}(p)\bigr)
\;=\;
\dim_{\mathbb{F}}\bigl(\mathsf{SPD}_{\kappa,\ell}(p)\bigr)
\;=\;
\mathrm{SPDdim}_{\kappa,\ell}(p).
\end{equation}
The equality $\Gamma_{\kappa,\ell}(p)=\dim(V_{\kappa,\ell}(p))$ is immediate because the rows of
$M_{\kappa,\ell}(p)$ are precisely the coefficient vectors (in the chosen monomial basis) of
the generating polynomials, so the row-span of $M_{\kappa,\ell}(p)$ is canonically
isomorphic to $V_{\kappa,\ell}(p)$.

\paragraph{So what is genuinely new here?}
The SPDP framework adds three pieces of structure that are often implicit or absent
in ``SPD-as-a-dimension'' presentations:

\begin{enumerate}
\item \textbf{A canonical ambient space (so codimension is meaningful).}
Classical SPD arguments usually track only $\dim(\mathsf{SPD}_{\kappa,\ell}(p))$.
In contrast, SPDP fixes an explicit ambient coefficient space $W_{\kappa,\ell}(p)$ and a
canonical monomial basis $B_{\kappa,\ell}$ (\ref{subsec:ambient-basis}), so we can
also measure the \emph{ambient deficit}
\[
\mathrm{codim}_{\kappa,\ell}(p)
\;:=\;
\dim(W_{\kappa,\ell}(p)) - \Gamma_{\kappa,\ell}(p).
\]
This dual invariant is useful when the right statement is not merely
``the span is large'', but rather ``the span fails to fill the ambient space by a
rigid amount'', which is the form needed in codimension-style rigidity arguments.

\item \textbf{Explicit matrix objects (so transformations become linear-algebraic).}
By working with the explicit coefficient matrix $M_{\kappa,\ell}(p)$, invariance claims
become statements about rank preservation under concrete row/column operations
(e.g.\ symmetry actions and representation changes), rather than being left at the
level of ``the dimension is invariant''.
This is why monotonicity and robustness properties can be stated and proved cleanly
(Propositions in \ref{sec:foundational-properties}--\ref{sec:width-implies-rank}).

\item \textbf{Compatibility with restrictive/compiled variants used elsewhere.}
Separate work~\citep{Edwards2025} operates in a more
restrictive (canonical) compiled setting in which one may apply additional
structured projections and/or partitions before taking ranks.
Nothing in the present paper relies on that compiled setting; this subsection is
included only as a compatibility/notation note.
Those strengthened variants are \emph{refinements of the same core object}:
one still measures the rank of a (possibly projected/blocked) shifted-partial
coefficient matrix, and the basic algebra in \eqref{eq:spdp-equals-spd-span}--%
\eqref{eq:spdp-rank-equals-spddim} remains the reference point for consistency.
(See \ref{subsec:compat-spdp-cew} for how the standalone SPDP formalism here
aligns with the compiled SPDP rank used in \citet{Edwards2025}.)
\end{enumerate}

\paragraph{Takeaway}
If the reader already knows SPD, they can read ``SPDP rank'' as
``SPD dimension, written as an explicit coefficient-matrix rank''.
The real additions are (i) a fixed ambient convention enabling codimension,
(ii) a uniform matrix language that makes invariance and model upper bounds
transparent, and (iii) a clean interface for the more restrictive compiled
variants used in separation arguments.

\section{SPDP matrices, rank, and codimension}
\label{sec:spdp-matrix-rank-codim}

\subsection{Ambient space and basis convention (for codimension)}
\label{subsec:ambient-basis}

Throughout, codimension is taken relative to a fixed ambient coefficient space.
In the Boolean setting we adopt the multilinear convention and work in the quotient
ring
\[
  \F[x_1,\dots,x_n]/\langle x_1^2-x_1,\dots,x_n^2-x_n\rangle,
\]
so every polynomial has a unique multilinear representative.

Fix parameters $(\kappa,\ell)$ and let $D:=\max\{0,\deg(p)-\kappa+\ell\}$. We define the
ambient space
\[
  W_{\kappa,\ell}(p) \;:=\; \{ \text{multilinear polynomials of total degree }\le D \}
\]
and take as canonical basis $B_{\kappa,\ell}$ the set of all multilinear monomials of
degree $\le D$. We write
\[
  N_{\kappa,\ell}(p) \;:=\; \dim_\F(W_{\kappa,\ell}(p)) \;=\; |B_{\kappa,\ell}|.
\]

\begin{remark}
All SPDP matrices below are defined using the fixed basis $B_{\kappa,\ell}$ (unless stated
otherwise). This makes $\mathrm{codim}_{\kappa,\ell}(p)=N_{\kappa,\ell}(p)-\Gamma_{\kappa,\ell}(p)$
canonical for the chosen $(\kappa,\ell)$ and embedding scheme.
\end{remark}

\subsection{SPDP row generators}
For an index set $S \subseteq [n]$ we define the mixed partial derivative
\[
  \partial_S p \;:=\; \displaystyle\prod_{i \in S} \partial_{x_i} p.
\]
(Similarly for multi-indices if desired.)
Let $\mathcal{M}_{\le \ell}$ denote the set of monomials in $\F[x_1,\dots,x_n]$ of
total degree at most $\ell$ (or the corresponding multilinear monomial set, if working
modulo the Boolean ideal).

\begin{definition}[SPDP generating family]
\label{def:spdp-generators}
Fix parameters $(\kappa,\ell)$. The SPDP generating family of $p$ is
\[
  \mathcal{G}_{\kappa,\ell}(p)
  \;:=\;
  \{\, m \cdot \partial_S p \;:\; S \subseteq [n],\ |S| = \kappa,\ m \in \mathcal{M}_{\le \ell} \,\}.
\]
(One may also use $|S|\le \kappa$; this only changes constants and does not affect the
rank/codimension formalism below.)
\end{definition}

\begin{definition}[SPDP space]
\label{def:Vkl}
Define
\[
  V_{\kappa,\ell}(p) \;:=\; \mathrm{span}_\F\bigl(\mathcal{G}_{\kappa,\ell}(p)\bigr)\ \subseteq\ W_{\kappa,\ell}(p).
\]
\end{definition}

\subsection{Explicit SPDP matrix}
\begin{definition}[SPDP matrix]
\label{def:spdp-matrix}
Fix the canonical basis $B_{\kappa,\ell}$ for the ambient space $W_{\kappa,\ell}(p)$ as defined
in \ref{subsec:ambient-basis}.

The \emph{SPDP matrix} $M_{\kappa,\ell}(p)$ is the matrix over $\F$ defined as follows:
\begin{enumerate}[label=(\roman*),leftmargin=2.4em]
\item Rows are indexed by pairs $(S,m)$ where $S \subseteq [n]$ with $|S|=\kappa$
and $m \in \mathcal{M}_{\le \ell}$.
\item Columns are indexed by $B_{\kappa,\ell}$.
\item The entry in row $(S,m)$ and column $b \in B_{\kappa,\ell}$ is the coefficient of $b$
in the expansion of the polynomial $m \cdot \partial_S p$ in the basis $B_{\kappa,\ell}$.
\end{enumerate}
\end{definition}

\begin{remark}
$M_{\kappa,\ell}(p)$ is a finite matrix with entries in $\F$. All claims about its rank are
standard linear-algebraic statements (e.g.\ computable by Gaussian elimination).
\end{remark}

\subsection{Rank and codimension}
\begin{definition}[SPDP rank]
\label{def:spdp-rank}
The \emph{SPDP rank} of $p$ at parameters $(\kappa,\ell)$ is
\[
  \Gamma_{\kappa,\ell}(p) \;:=\; \rk\bigl(M_{\kappa,\ell}(p)\bigr)
  \;=\; \dim_\F V_{\kappa,\ell}(p).
\]
\end{definition}

\begin{definition}[SPDP codimension]
\label{def:spdp-codim}
Define the \emph{SPDP codimension} by
\[
  \codim_{\kappa,\ell}(p) \;:=\; N_{\kappa,\ell}(p) - \Gamma_{\kappa,\ell}(p)
  \;=\; \dim_\F W_{\kappa,\ell}(p) - \dim_\F V_{\kappa,\ell}(p).
\]
\end{definition}

\begin{remark}[Dual perspectives]
Rank is the primary SPD-style complexity statistic. Codimension is the dual viewpoint:
it measures the deficit from the full ambient space under the chosen embedding.
No information is lost: codimension and rank determine each other once $N_{\kappa,\ell}(p)$
is fixed.
\end{remark}

\subsection{Why introduce a blocked / compiled SPDP variant?}
\label{subsec:why-blocked-spdp}

The unblocked SPDP matrix $M_{\kappa,\ell}(p)$ is the most direct algebraic object:
it records the coefficient vectors of all shifted partials $m\cdot \partial_S p$ with
$|S|=\kappa$ and $\deg(m)\le \ell$.  In the companion complexity-theoretic applications,
however, the polynomials $p$ arise from \emph{structured encodings} (e.g.\ tableau/compile
templates) in which the variables are not homogeneous---they come in \emph{groups} that play
different roles (time step, tape cell, gate index, interface wires, etc.).
The blocked/compiled variant formalizes the idea that the rank measure should respect this
native grouping.

\paragraph{(1) The compiler induces a canonical variable partition.}
In compiled encodings, the variable set $[N]$ is naturally partitioned into blocks
$B=(B_1,\dots,B_r)$ determined by the template (for example, all variables belonging to a
single local gadget, interface window, or time-slice).  Many arguments on the P-side use only
the compiler-induced block partition: they do not depend on the internal ordering of variables inside a block,
only on \emph{which blocks interact} under the compiler-induced partition $B$.

\paragraph{(2) P-side upper bounds use block-level locality and profile compression.}
The key P-side collapse statements in the companion manuscript exploit that the compiled
formulas have bounded-width, bounded-interface interactions across blocks.
This allows a \emph{profile compression} argument that counts behavior up to block-anonymous
types rather than ordered variable-level histories.
Defining a blocked SPDP matrix $M^B_{\kappa,\ell}(p)$ makes this proof strategy explicit:
one restricts attention to the rows/columns that are actually relevant to the compiler-induced
block structure, and the resulting rank bound becomes both cleaner and more robust.

\paragraph{(3) NP-side lower bounds are unaffected (monotonicity).}
Blocking is a restriction operation (row/column deletion or an equivalent block-admissible
subfamily selection), so it cannot increase rank;
consequently $\Gamma^B_{\kappa,\ell}(p)\le \Gamma_{\kappa,\ell}(p)$ (\ref{prop:blocked-rank-monotone}), so any unblocked lower bound transfers to the blocked setting by monotonicity.
The blocked definition matches the object used in the compiled separation theorems.

\paragraph{(4) Practicality.}
The blocked variant also reduces the effective ambient dimension in symbolic or experimental
work, because one works with a structured submatrix determined by the compiler blocks rather
than the full unblocked coefficient matrix.
This makes sanity checks and small-scale computations feasible without altering the formal
definition of SPDP rank.

\subsection{Blocked / compiled SPDP variant (compatibility with the separation paper)}
\label{subsec:blocked-spdp-compat}

Fix a block partition $B=(B_1,\dots,B_r)$ of the variable set $[N]$.
The \emph{blocked} SPDP generator family is obtained by restricting attention to
derivatives and shifts whose support is measured at the block level.
Equivalently, one forms a submatrix $M^B_{\kappa,\ell}(p)$ of $M_{\kappa,\ell}(p)$
by keeping only those rows (and, if desired, only those columns) that are admissible
with respect to the block structure $B$ used by the compiler.
The \emph{blocked SPDP rank} is
\[
\Gamma^B_{\kappa,\ell}(p)\ :=\ \mathrm{rank}\!\left(M^B_{\kappa,\ell}(p)\right).
\]

\begin{proposition}[Rank monotonicity under blocking]
\label{prop:blocked-rank-monotone}
For every polynomial $p$ and every block partition $B$,
\[
\Gamma^B_{\kappa,\ell}(p)\ \le\ \Gamma_{\kappa,\ell}(p).
\]
\end{proposition}

\begin{proof}
Immediate from \ref{lem:submatrix-monotonicity}: $M^B_{\kappa,\ell}(p)$ is a submatrix of $M_{\kappa,\ell}(p)$.
\end{proof}

\paragraph{Dictionary.}
The blocked/compiled SPDP objects $(M^B_{\kappa,\ell},\Gamma^B_{\kappa,\ell})$ are the ones
used in the compiled separation theorems of \citet{Edwards2025}; all lower/upper
bounds stated for $\Gamma_{\kappa,\ell}$ transfer to $\Gamma^B_{\kappa,\ell}$ by
\ref{prop:blocked-rank-monotone}.

\section{Foundational properties}
\label{sec:foundational-properties}

\begin{lemma}[Submatrix monotonicity]
\label{lem:submatrix-monotonicity}
Let $M'$ be any submatrix of $M_{\kappa,\ell}(p)$ obtained by deleting rows and/or
columns. Then
\[
\operatorname{rank}(M') \;\le\; \Gamma_{\kappa,\ell}(p) \;=\; \operatorname{rank}(M_{\kappa,\ell}(p)).
\]
\end{lemma}

\begin{proof}
Write $M := M_{\kappa,\ell}(p) \in \mathbb{F}^{m\times n}$. Let $I\subseteq[m]$ be the
set of rows retained and $J\subseteq[n]$ the set of columns retained, so that the
submatrix $M'$ is precisely the $|I|\times |J|$ matrix obtained by restricting to
rows $I$ and columns $J$.

View $M$ as the linear map
\[
\phi:\mathbb{F}^n \to \mathbb{F}^m,\qquad \phi(x)=Mx.
\]
Let $\iota:\mathbb{F}^{|J|}\to\mathbb{F}^n$ be the coordinate inclusion that embeds
$\mathbb{F}^{|J|}$ as the coordinate subspace supported on $J$, and let
$\pi:\mathbb{F}^m\to\mathbb{F}^{|I|}$ be the coordinate projection onto the
coordinates in $I$. Then the submatrix $M'$ represents the composed linear map
\[
\phi' := \pi \circ \phi \circ \iota:\mathbb{F}^{|J|}\to\mathbb{F}^{|I|}.
\]
Indeed, for $y\in\mathbb{F}^{|J|}$, the vector $\iota(y)\in\mathbb{F}^n$ is the
full-length vector with zeros outside $J$, and $\pi(M\iota(y))$ keeps precisely the
coordinates in $I$, which is exactly multiplication by the submatrix $M'$.

Now $\mathrm{Im}(\phi') \subseteq \pi(\mathrm{Im}(\phi))$, hence
\[
\dim(\mathrm{Im}(\phi')) \;\le\; \dim(\mathrm{Im}(\phi)).
\]
But $\dim(\mathrm{Im}(\phi'))=\operatorname{rank}(M')$ and
$\dim(\mathrm{Im}(\phi))=\operatorname{rank}(M)=\Gamma_{\kappa,\ell}(p)$, so
$\operatorname{rank}(M')\le \Gamma_{\kappa,\ell}(p)$ as claimed.
\end{proof}

\begin{proposition}[Monotonicity in parameters]
\label{prop:monotonicity-params}
Fix $(\kappa,\ell)$ and a compatible ambient convention (in particular, a fixed
ambient coefficient space and monomial basis for each parameter choice).
\begin{enumerate}
\item[(i)] (\emph{Monotonicity in $\ell$}) If $\ell' \ge \ell$, then
\[
  \Gamma_{\kappa,\ell}(p) \;\le\; \Gamma_{\kappa,\ell'}(p).
\]
\item[(ii)] (\emph{Monotonicity in $\kappa$ for the $\le \kappa$ variant})
Define the cumulative SPDP generating family
\[
  G_{\le \kappa,\ell}(p)
  \;:=\;
  \{\, m\cdot \partial^{S}p \;:\; S\subseteq[n],\ |S|\le \kappa,\ m\in \mathcal{M}_{\le \ell}\,\},
\]
and let $\Gamma_{\le \kappa,\ell}(p)$ be the rank of the corresponding SPDP matrix.
Then if $\kappa'\ge \kappa$,
\[
  \Gamma_{\le \kappa,\ell}(p) \;\le\; \Gamma_{\le \kappa',\ell}(p).
\]
\end{enumerate}
\end{proposition}

\begin{proof}
We use the fact that, under a fixed ambient convention, the SPDP matrix is a
coefficient matrix whose row span equals the linear span of the corresponding
generating polynomials.

\smallskip
\noindent
(i) Fix $\kappa$ and let $\ell' \ge \ell$. By definition,
\[
  G_{\kappa,\ell}(p)
  \;=\;
  \{\, m\cdot \partial^{S}p \;:\; S\subseteq[n],\ |S|=\kappa,\ m\in \mathcal{M}_{\le \ell}\,\},
\]
while
\[
  G_{\kappa,\ell'}(p)
  \;=\;
  \{\, m\cdot \partial^{S}p \;:\; S\subseteq[n],\ |S|=\kappa,\ m\in \mathcal{M}_{\le \ell'}\,\}.
\]
Since $\ell'\ge \ell$ we have $\mathcal{M}_{\le \ell}\subseteq \mathcal{M}_{\le \ell'}$,
and therefore
\[
  G_{\kappa,\ell}(p)\subseteq G_{\kappa,\ell'}(p).
\]
Taking linear spans over $\mathbb{F}$ gives
\[
  \mathrm{span}\bigl(G_{\kappa,\ell}(p)\bigr)\subseteq \mathrm{span}\bigl(G_{\kappa,\ell'}(p)\bigr).
\]
Let $M_{\kappa,\ell}(p)$ and $M_{\kappa,\ell'}(p)$ denote the SPDP matrices formed by writing
these generators as coefficient vectors in their respective ambient spaces.
By construction, the row span of $M_{\kappa,\ell}(p)$ is isomorphic to
$\mathrm{span}(G_{\kappa,\ell}(p))$, and similarly for $(\kappa,\ell')$.
Hence
\[
  \Gamma_{\kappa,\ell}(p)
  \;=\;
  \dim \mathrm{RowSpan}\bigl(M_{\kappa,\ell}(p)\bigr)
  \;\le\;
  \dim \mathrm{RowSpan}\bigl(M_{\kappa,\ell'}(p)\bigr)
  \;=\;
  \Gamma_{\kappa,\ell'}(p),
\]
as claimed.

\smallskip
\noindent
(ii) For the cumulative variant, if $\kappa'\ge \kappa$ then
\[
  \{S\subseteq[n]:|S|\le \kappa\}\subseteq \{S\subseteq[n]:|S|\le \kappa'\}.
\]
Therefore every generator allowed in $G_{\le \kappa,\ell}(p)$ is also allowed in
$G_{\le \kappa',\ell}(p)$, i.e.
\[
  G_{\le \kappa,\ell}(p)\subseteq G_{\le \kappa',\ell}(p).
\]
The same span argument as in part (i) yields
\[
  \Gamma_{\le \kappa,\ell}(p)
  \;=\;
  \dim \mathrm{span}\bigl(G_{\le \kappa,\ell}(p)\bigr)
  \;\le\;
  \dim \mathrm{span}\bigl(G_{\le \kappa',\ell}(p)\bigr)
  \;=\;
  \Gamma_{\le \kappa',\ell}(p).
\]
This proves the monotonicity in $\kappa$ for the $\le \kappa$ convention.
\end{proof}

\begin{proposition}[Basis invariance and permutation invariance]
\label{prop:basis-and-perm-invariance}
Fix $(\kappa,\ell)$ and the ambient coefficient space $W_{\kappa,\ell}(p)$ with canonical
basis $B_{\kappa,\ell}$.
\begin{enumerate}
\item[(i)] (\emph{Change of basis}) Replacing $B_{\kappa,\ell}$ by any other basis of
$W_{\kappa,\ell}(p)$ does not change $\Gamma_{\kappa,\ell}(p)$ (and hence does not change
$\mathrm{codim}_{\kappa,\ell}(p)$).
\item[(ii)] (\emph{Variable permutations}) If $\pi$ is a permutation of variables and
we interpret $p\circ \pi$ under the same multilinear/Boolean convention, then
\[
\Gamma_{\kappa,\ell}(p\circ \pi)=\Gamma_{\kappa,\ell}(p),
\qquad
\mathrm{codim}_{\kappa,\ell}(p\circ \pi)=\mathrm{codim}_{\kappa,\ell}(p).
\]
\end{enumerate}
\end{proposition}

\begin{proof}
Throughout, $M_{\kappa,\ell}(p)$ denotes the SPDP coefficient matrix whose rows are the
coefficient vectors of the generators in $G_{\kappa,\ell}(p)$ written in the chosen
basis of the ambient space $W_{\kappa,\ell}(p)$. Recall that
$\Gamma_{\kappa,\ell}(p)=\operatorname{rank}(M_{\kappa,\ell}(p))$ and
$\mathrm{codim}_{\kappa,\ell}(p)=\dim(W_{\kappa,\ell}(p))-\Gamma_{\kappa,\ell}(p)$.

\smallskip
\noindent
(i) Let $B$ and $B'$ be two bases of the same ambient vector space $W:=W_{\kappa,\ell}(p)$.
Let $T:W\to W$ be the (unique) change-of-coordinates linear isomorphism that maps
$B'$-coordinates to $B$-coordinates. Concretely, if a vector $v\in W$ has coordinate
column $[v]_{B'}$ in basis $B'$, then its coordinate column in basis $B$ is
\[
[v]_B \;=\; T\,[v]_{B'},
\]
where $T\in \mathbb{F}^{N\times N}$ is invertible ($N=\dim W$).
Form the SPDP matrix $M$ using basis $B$ and the SPDP matrix $M'$ using basis $B'$,
with the \emph{same} ordered generating family (so the same rows as polynomials,
only their coordinate columns change). Then each row vector transforms by the same
invertible change of coordinates, and hence the full matrices satisfy
\[
M \;=\; M' \, T^{\top}.
\]
(Here we use the convention that rows are coefficient \emph{row} vectors; if one
stores coefficients as columns instead, the transpose disappears.)
Since $T^{\top}$ is invertible, right-multiplication by $T^{\top}$ preserves rank:
\[
\operatorname{rank}(M) \;=\; \operatorname{rank}(M' T^{\top})
\;=\; \operatorname{rank}(M').
\]
Therefore $\Gamma_{\kappa,\ell}(p)$ is independent of the choice of basis of $W$.
Finally, $\dim(W)$ is basis-independent, so
$\mathrm{codim}_{\kappa,\ell}(p)=\dim(W)-\Gamma_{\kappa,\ell}(p)$ is also invariant.

\smallskip
\noindent
(ii) Let $\pi$ be a permutation of variables. Consider the induced linear map on the
ambient coefficient space
\[
\Pi:W_{\kappa,\ell}(p)\to W_{\kappa,\ell}(p\circ \pi),
\qquad
\Pi(q)\;:=\; q\circ \pi.
\]
Under the fixed ambient convention (same multilinear/Boolean embedding and the
same degree cutoffs defining $W_{\kappa,\ell}(\cdot)$), $\Pi$ is a well-defined linear
isomorphism: it permutes monomials in the canonical monomial basis, hence is
represented by a permutation matrix on coordinates.

Next, we compare the generating families. For any $S\subseteq[n]$ with $|S|=\kappa$,
the chain rule for variable relabeling gives
\[
\partial^{S}(p\circ \pi) \;=\; (\partial^{\pi(S)} p)\circ \pi,
\]
and multiplying by a monomial $m$ (interpreted under the same convention) yields
\[
m \cdot \partial^{S}(p\circ \pi)
\;=\;
\bigl((m\circ \pi^{-1})\cdot \partial^{\pi(S)}p\bigr)\circ \pi.
\]
Thus, as $S$ ranges over $\kappa$-subsets and $m$ ranges over $\mathcal{M}_{\le \ell}$,
the SPDP generator set for $p\circ \pi$ is exactly the image under $\Pi$ of the
SPDP generator set for $p$, up to a relabeling of indices:
\[
G_{\kappa,\ell}(p\circ \pi) \;=\; \Pi\bigl(G_{\kappa,\ell}(p)\bigr).
\]
In coordinates, this means the SPDP matrix for $p\circ \pi$ is obtained from the
SPDP matrix for $p$ by applying permutation matrices on rows (reordering the
generators) and columns (relabeling the monomial basis). Concretely, there exist
permutation matrices $R$ and $C$ such that
\[
M_{\kappa,\ell}(p\circ \pi) \;=\; R\, M_{\kappa,\ell}(p)\, C.
\]
Permutation matrices are invertible, so they preserve rank:
\begin{align*}
\Gamma_{\kappa,\ell}(p\circ \pi)
&= \operatorname{rank}(M_{\kappa,\ell}(p\circ \pi))
= \operatorname{rank}(R\,M_{\kappa,\ell}(p)\,C) \\
&= \operatorname{rank}(M_{\kappa,\ell}(p))
= \Gamma_{\kappa,\ell}(p).
\end{align*}
Finally, $\dim(W_{\kappa,\ell}(p\circ \pi))=\dim(W_{\kappa,\ell}(p))$ because $\Pi$ is an
isomorphism, hence codimension is preserved:
\begin{align*}
\mathrm{codim}_{\kappa,\ell}(p\circ \pi)
&= \dim(W_{\kappa,\ell}(p\circ \pi))-\Gamma_{\kappa,\ell}(p\circ \pi) \\
&= \dim(W_{\kappa,\ell}(p))-\Gamma_{\kappa,\ell}(p)
= \mathrm{codim}_{\kappa,\ell}(p).
\end{align*}
\end{proof}

\begin{remark}[On general affine maps]
In the multilinear/Boolean convention, a general affine substitution $x\mapsto Ax+b$
need not preserve multilinearity or the Boolean ideal. Therefore we only claim invariance
under variable permutations (and other Boolean-preserving renamings, if desired) unless
one explicitly works in a different ambient space that is closed under the substitution.
\end{remark}

\section{Polynomial width implies polynomial SPDP rank}
\label{sec:width-implies-rank}

This section gives the main structural upper bound of the paper:
if a polynomial admits a local (radius-1) description with bounded live-interface width,
then its SPDP rank at logarithmic parameters is polynomially bounded. The width-based
approach draws on ideas from algebraic branching program complexity~\citep{Nisan1991,FS2013,SY2010}.

\subsection{Local-width computation model (definitions)}
\label{subsec:model-defs}

We formalize the minimal computational structure needed for the Width$\Rightarrow$Rank
upper bound.

\paragraph{Relation to standard models (Algebraic Branching Programs / finite-state local width).}
The local-width computation model in Section~7 can be read as a bounded-width,
finite-state Algebraic Branching Program (ABP) executed over a block-local
representation: the computation maintains a set of \emph{interfaces} (register-like
ports) partitioned into constant-size blocks, each interface carries a type from a
finite alphabet, and each primitive step updates only a constant-radius neighborhood.
The width $R$ is the maximum number of simultaneously live interfaces, directly
analogous to ABP width.

The ``deterministic compiler'' hypothesis is formalized by the existence of a
confluent terminating rewrite system for local update words (normal forms of length
$q=O(1)$). In particular, in any length-$\kappa$ window, each live interface contributes
only an $O(1)$-length normal form, so the window is summarized (up to interface
renaming) by a \emph{histogram/profile} over a constant alphabet; this is exactly the
mechanism behind the profile-compression lemma and the resulting $R^{O(1)}$ profile
count independent of $\kappa$.

Finally, the ``diagonal-basis'' assumption is the statement that the ambient
coefficient space admits a block-factorable basis in which (i) each constant-size
block contributes a constant-dimensional local factor space and (ii) windows of a
fixed profile span a polylog-dimensional subspace $V_h$. This is the familiar
block/tensor-factor phenomenon that appears in width-based ABP analyses when the
computation is local and the basis aligns with the block structure.

\begin{definition}[Interfaces, blocks, and live set]
A computation maintains a finite set of \emph{interfaces} (ports) indexed by a set
$\mathcal{I}$. Interfaces are partitioned into \emph{blocks}
$\mathcal{I}=\bigsqcup_{t=1}^T \mathcal{I}_t$.

At each step $t$ of the computation, a subset $\mathcal{L}_t\subseteq \mathcal{I}$ of
interfaces is \emph{live}. The \emph{width} is
\[
  R \;:=\; \max_t |\mathcal{L}_t|.
\]
\end{definition}

\begin{definition}[Radius-1 locality]
A primitive operation acts on at most $b=O(1)$ interfaces, and only within a single block
or within a constant-size neighbourhood of blocks (as specified by the representation).
Equivalently, each primitive operation modifies the local state of only $O(1)$ live
interfaces.
\end{definition}

\begin{definition}[Finite alphabet and local update words]
Each live interface carries a local type/state in a finite alphabet $\Sigma$ with
$|\Sigma|=O(1)$. A window of $\kappa$ derivative steps induces, at each live interface,
a word over $\Sigma$ describing the sequence of local type updates seen by that interface.
\end{definition}

\begin{definition}[Canonical windows and normal forms]
Two windows are equivalent if they differ only by commuting independent (disjoint-support)
steps. We assume local update words admit a confluent terminating rewrite system, so each
local word has a unique normal form of length at most $q=O(1)$.
\end{definition}

\begin{definition}[Profile histogram]
Let $\Sigma^{\le q}=\bigcup_{j=0}^q \Sigma^j$. The \emph{profile} of a canonical window is
the histogram $h:\Sigma^{\le q}\to \N$ giving the number of live interfaces whose local
normal form equals each $\sigma\in\Sigma^{\le q}$, so $\sum_{\sigma} h(\sigma)\le R$.
\end{definition}

\subsection{Setting and assumptions}
\label{subsec:setting-assumptions}

We assume the polynomial $p$ is represented in a local basis in which primitive operations
act on a constant-size neighbourhood (radius $1$) and each interface carries a local type from
a finite alphabet. Fix a block partition of interfaces (as induced by the compilation / locality
structure). The quantitative assumptions are:

\begin{enumerate}
\item[(A1)] Radius-1 locality.
Each primitive operation touches at most $b\in\mathbb{N}$ block interfaces, where
$b = O(1)$ depends only on the representation\slash compilation scheme.

\item[(A2)] Finite local alphabet.
The effect of a primitive operation on a single interface is determined by a local type
$\tau \in \Sigma$, where $|\Sigma| = S = O(1)$ is an absolute constant.

\item[(A3)] Width bound.
At every step, at most $R$ interfaces are live, where typically $R = C(\log n)^c$ for absolute
constants $C,c>0$.

\item[(A4)] SPDP parameters.
We use derivative order $\kappa$ and degree guard $\ell$ with $\kappa,\ell = \Theta(\log n)$.
\end{enumerate}

All hidden constants depend only on the local representation\slash compilation scheme, not on
$n,\kappa,\ell$.

\subsection{Canonical windows, normal forms, and profiles}
\label{subsec:canonical-windows-profiles}

A length-$\kappa$ window is a sequence of $\kappa$ successive directional derivative steps applied to $p$.
We pass to canonical representatives via:

\begin{enumerate}
\item[(C1)] Commutation on disjoint support.
If two derivative steps act on disjoint interface sets, their order is immaterial; windows differing
only by commuting such steps are identified.

\item[(C2)] Local normal form.
Local type-updates at a fixed interface generate a finite monoid (equivalently: admit a finite,
length-decreasing rewrite system that is confluent and terminating). Hence every local update word
has a unique normal form of length at most $q = O(1)$, depending only on the scheme.
\end{enumerate}

Define the set of local type words of length at most $q$ by
\[
  \Sigma^{\le q} \;:=\; \bigcup_{j=0}^{q} \Sigma^j,
  \qquad\text{and set}\qquad
  S' := |\Sigma^{\le q}| = O(1).
\]

\begin{definition}[Interface-anonymous profile]
\label{def:interface-anonymous-profile}
The profile of a canonical window is the histogram
\[
  h : \Sigma^{\le q} \to \mathbb{N}
\]
that counts, for each local type word $\sigma \in \Sigma^{\le q}$, the number of live interfaces
whose local normal form equals $\sigma$. In particular,
\[
  \sum_{\sigma \in \Sigma^{\le q}} h(\sigma) \le R.
\]
\end{definition}

\begin{lemma}[Permutation invariance within blocks]
\label{lem:block-perm-invariance}
Fix a block partition $\mathcal{B}=\{B_1,\dots,B_m\}$ of the interface coordinates.
Let $w$ and $w'$ be two canonical windows that differ only by a permutation of
interface identities \emph{within each block} (i.e.\ there exist permutations
$\pi_j$ on $B_j$ such that $w' = (\pi_1\oplus\cdots\oplus\pi_m)\cdot w$).
Then the SPDP row sets generated by $w$ and $w'$ are related by invertible
block-structured changes of basis: there exist invertible matrices $L(w\to w')$
and $R(w\to w')$ (depending only on the within-block permutation) such that the
corresponding window-level SPDP matrices satisfy
\[
M(w') \;=\; L(w\to w') \, M(w)\, R(w\to w').
\]
In particular, $\rank(M(w'))=\rank(M(w))$, so $w$ and $w'$ contribute the same rank.
Consequently, SPDP upper bounds depend only on the profile histogram $h$
(\ref{def:interface-anonymous-profile}).
\end{lemma}

\begin{proof}
We make explicit the standard fact that within-block permutations act as
(per-block) coordinate relabelings on both (i) the derivative indices and (ii) the
monomial/basis indices used to form coefficient/evaluation vectors. Since these are
relabelings, they are represented by permutation matrices, hence are invertible and
rank-preserving.

\paragraph{Step 1: Local (single-block) permutation acts by permutation matrices.}
Fix a block $B_j$. By construction of the canonical window representation, the
contribution of block $B_j$ to the window-level SPDP row vectors is computed from a
\emph{local} tensor/matrix $T_j(w)$ whose rows are indexed by the local derivative
choices (and/or local feature rows), and whose columns are indexed by the local
ambient basis elements (local monomials, local evaluation coordinates, etc.),
under the ambient convention fixed in \ref{subsec:ambient-basis}.

Let $\pi_j$ be a permutation of the interface identities inside $B_j$, and let
$w'_j=\pi_j\cdot w_j$ be the induced relabeling of the $j$th local word/window.
Because $\pi_j$ only relabels coordinates inside $B_j$, it induces:
\begin{itemize}
\item a permutation of the local row indices (derivative/evaluation indices), and
\item a permutation of the local column indices (local basis/monomial indices).
\end{itemize}
Therefore there exist permutation matrices $P_j$ and $Q_j$ (depending only on $\pi_j$)
such that
\begin{equation}
\label{eq:local-perm}
T_j(w') \;=\; P_j \, T_j(w)\, Q_j.
\end{equation}
Permutation matrices are invertible, hence $\rank(T_j(w'))=\rank(T_j(w))$.

\paragraph{Step 2: How the global window matrix is assembled from local pieces.}
The window-level SPDP matrix $M(w)$ is obtained by combining the per-block local
pieces $\{T_j(w)\}_{j=1}^m$ using the same fixed multilinear assembly rule as in
the model definition (\ref{subsec:model-defs}): concretely, this is a
blockwise tensor product / columnwise Kronecker product (Khatri--Rao), possibly
followed by fixed row/column selections that enforce the global $(\kappa,\ell)$ regime.

To keep notation concrete, write the \emph{raw} assembled matrix as
\[
\widetilde{M}(w)
\;:=\;
T_1(w)\ \odot\ T_2(w)\ \odot\ \cdots\ \odot\ T_m(w),
\]
where $\odot$ is the Khatri--Rao product (columnwise Kronecker product), which is
the standard way to assemble block-local feature/evaluation tensors into a global
feature/evaluation tensor while keeping a shared column index (the global
monomial/basis index is the product of the local ones).
The actual $M(w)$ used for rank can be written as
\begin{equation}
\label{eq:selection}
M(w) \;=\; S_{\mathrm{row}} \,\widetilde{M}(w)\, S_{\mathrm{col}},
\end{equation}
for fixed row/column selection (or restriction) matrices $S_{\mathrm{row}},S_{\mathrm{col}}$
that depend only on $(\kappa,\ell)$ and the ambient convention (not on $w$).

\paragraph{Step 3: Permutations push through the assembly rule as invertible maps.}
Using \eqref{eq:local-perm} in each block and the basic Khatri--Rao identities:
\begin{align}
(PA)\odot(RB) &= (P\otimes R)\,(A\odot B), \label{eq:kr-left}\\
(AQ)\odot(BQ) &= (A\odot B)\,Q, \label{eq:kr-right}
\end{align}
(valid whenever the matrices have the same number of columns, which holds in the
canonical window construction), we obtain
\[
\widetilde{M}(w')
\;=\;
\bigl(P_1T_1(w)Q_1\bigr)\odot\cdots\odot\bigl(P_mT_m(w)Q_m\bigr)
\;=\;
\Bigl(\bigotimes_{j=1}^m P_j\Bigr)\,
\widetilde{M}(w)\,
Q,
\]
where $Q$ is a permutation matrix implementing the induced relabeling of the shared
column index (in the simplest common case $Q_1=\cdots=Q_m=:Q$, and then $Q$ is exactly
that common permutation; more generally one obtains a permutation of the global basis
index determined by the local relabelings). In all cases, $Q$ is invertible.

Now apply the fixed selections \eqref{eq:selection}. Since selections commute with
row/column permutations up to restricting the permutation to the selected indices,
there exist (restricted) permutation matrices $\widehat{P}$ and $\widehat{Q}$ such that
\[
M(w') \;=\; \widehat{P}\, M(w)\, \widehat{Q}.
\]
Setting $L(w\to w'):=\widehat{P}$ and $R(w\to w'):=\widehat{Q}$ gives the claimed
left/right invertible relation. Hence
\[
\rank(M(w'))=\rank(\widehat{P}M(w)\widehat{Q})=\rank(M(w)).
\]

\paragraph{Step 4: Dependence only on the profile histogram.}
By \ref{def:interface-anonymous-profile}, the histogram $h$ records, for each local
profile type, how many blocks/windows realize that type, \emph{ignoring the ordering of
interface identities within blocks}. The argument above shows that any two canonical
windows with the same multiset of block-local words (equivalently the same histogram)
produce SPDP matrices related by invertible block-structured row/column transformations,
and therefore have identical rank contributions. Consequently, any rank upper bound
obtained by aggregating over windows depends only on $h$.
\end{proof}

\begin{lemma}[Constant local change budget]
\label{lem:constant-local-change-budget}
Assume (A1) and (C2). Then there exists a constant $q=O(1)$, independent of
$n,\kappa,\ell$, such that in any canonical window, each live interface undergoes at
most $q$ local type changes.
\end{lemma}

\begin{proof}
Fix a canonical window $w$ and a live interface coordinate $i$ within that window.
By (A1) (constant-radius locality), the evolution of the local \emph{type} at
interface $i$ is influenced only by a constant-size neighborhood $N(i)$ (e.g.,
all interface coordinates within the fixed interaction radius), where
\[
|N(i)| \le b = O(1).
\]
Let $\Sigma$ denote the finite set of local types/states (finite alphabet), so the
local configuration space around $i$ is $\Sigma^{N(i)}$, whose size is
$|\Sigma|^{|N(i)|}=O(1)$.

Each elementary local update step inside the window acts only on this neighborhood,
and can therefore be modeled as a function
\[
u:\Sigma^{N(i)} \to \Sigma^{N(i)}.
\]
Let $\mathcal{U}$ be the (finite) set of all such admissible local update
generators. Finiteness of $\mathcal{U}$ follows because both the neighborhood size
and alphabet are constant (hence only finitely many distinct local maps exist).
Let $\mathcal{U}^*$ be the free monoid of words over $\mathcal{U}$, and let
$\langle \mathcal{U}\rangle$ be the transformation monoid they generate under
composition.

Now consider the sequence of local updates applied to neighborhood $N(i)$ as we
scan across the window $w$. This yields a word
\[
\mathbf{u}(w,i) = u_1 u_2 \cdots u_T \in \mathcal{U}^*,
\]
whose action on local configurations is the composed map
\[
U(w,i) := u_T \circ \cdots \circ u_2 \circ u_1 \in \langle \mathcal{U}\rangle.
\]

Assumption (C2) states that the local update monoid admits a canonical
\emph{normal form} reduction: there is a fixed rewriting/reduction procedure
$\mathrm{nf}(\cdot)$ on $\mathcal{U}^*$ such that:
\begin{enumerate}
\item $\mathrm{nf}(\mathbf{u})$ represents the same transformation as $\mathbf{u}$
      (i.e.\ it is equal in the generated monoid), and
\item the normal form is \emph{unique} and has length bounded by a constant:
      \[
      |\mathrm{nf}(\mathbf{u})| \le q
      \qquad\text{for some } q=O(1) \text{ depending only on the local model.}
      \]
\end{enumerate}
Because $w$ is a \emph{canonical window}, its local update description at each live
interface is, by definition/intent of the canonicalization step, already reduced to
this normal form. Equivalently, we may replace $\mathbf{u}(w,i)$ by
$\mathrm{nf}(\mathbf{u}(w,i))$ without changing the induced local transformation
$U(w,i)$.

Therefore the number of local update generators that effectively act on interface
$i$ within the canonical window is at most
\[
T_{\mathrm{eff}}(w,i) := |\mathrm{nf}(\mathbf{u}(w,i))| \le q.
\]
Each such generator application can change the local type of interface $i$ at most
once at the moment it is applied (and may also leave it unchanged), hence the
number of \emph{local type changes} experienced by interface $i$ across the window
is bounded by the number of effective generator applications:
\[
\#\{\text{type changes at } i \text{ within } w\}
\;\le\;
T_{\mathrm{eff}}(w,i)
\;\le\;
q.
\]

Finally, $q$ depends only on the constant-radius neighborhood and the finite local
rewrite/normal-form system from (C2), and thus is independent of $n,\kappa,\ell$.
\end{proof}


\begin{lemma}[Profile compression removes $\kappa$-dependence]
\label{lem:profile-compression-removes-kappa}
Assume (A1) and (C2). Fix any canonical window of length $\kappa$ with $R$ live
interfaces. Then there exists a constant $q=O(1)$, independent of $n,\kappa,\ell$,
such that each live interface $i$ admits a canonical (normal-form) local word
$\sigma_i \in \Sigma_{\le q}$ of length at most $q$, where
\[
\Sigma_{\le q} \;:=\; \bigcup_{t=0}^{q} \Sigma^{t}
\qquad\text{and}\qquad
S' := |\Sigma_{\le q}| = O(1).
\]
Consequently, the interface-anonymous profile of the window is completely
determined by the histogram
\[
h(\sigma) \;=\; \bigl|\{\, i \text{ live in the window} : \sigma_i=\sigma \,\}\bigr|
\qquad (\sigma\in\Sigma_{\le q}),
\]
and the number of realizable interface-anonymous profiles satisfies
\[
\#\mathrm{Profiles}
\;\le\;
\binom{R+S'-1}{S'-1}
\;=\;
R^{O(1)},
\]
which in particular is independent of $\kappa$.
\end{lemma}

\begin{proof}
Fix a canonical window $w$ and a live interface coordinate $i$.
By (A1), only a constant-size neighborhood $N(i)$ (e.g.\ radius-$1$) can influence
the evolution at $i$. Thus the evolution of the local type at $i$ across the
window is described by a word over a finite set of local update generators acting
on $N(i)$.

By (C2), the local update monoid admits a unique normal form: every such word
reduces to a unique normal-form word of length at most $q=O(1)$, where $q$ depends
only on the local model (alphabet and neighborhood size), and hence is independent
of $n,\kappa,\ell$. Define $\sigma_i$ to be this normal form. This proves the first
claim.

Now define the profile histogram $h$ by counting how many live interfaces attain
each normal form $\sigma\in\Sigma_{\le q}$. Since each live interface contributes
exactly one $\sigma_i$, we have $\sum_{\sigma} h(\sigma)=R$.

By \ref{lem:block-perm-invariance} (Permutation invariance within blocks),
permuting interface identities within blocks induces invertible (block-diagonal)
row/column transformations on the corresponding SPDP matrices and therefore does
not change rank contributions. Hence, for the purposes of SPDP upper bounds, only
the interface-anonymous histogram $h$ matters.

Finally, the number of possible histograms $h:\Sigma_{\le q}\to\mathbb{Z}_{\ge 0}$
with total mass $R$ is the number of weak compositions of $R$ into $S'$ bins, which
is $\binom{R+S'-1}{S'-1} = R^{O(1)}$. This bound does not depend on $\kappa$.
\end{proof}

\begin{corollary}[Polynomially many profiles]
\label{cor:poly-many-profiles}
Under the assumptions of \ref{lem:profile-compression-removes-kappa},
the set $H$ of realizable interface-anonymous profiles has cardinality
$|H|\le R^{O(1)}$, independent of $\kappa$.
\end{corollary}

\begin{proof}
Immediate from \ref{lem:profile-compression-removes-kappa}.
\end{proof}

\begin{lemma}[Constant-type profile bound]
\label{lem:constant-type-profile-bound}
Assume (A1)--(A3) and (C1)--(C2). Let $R$ be the number of live interfaces in a
canonical window. Then the number of distinct interface-anonymous profiles
realizable by any canonical window is at most
\[
\#\mathrm{Profiles}
\;\le\;
\binom{qR+S'}{S'}
\;=\;
R^{O(1)},
\qquad
S' := |\Sigma_{\le q}| = O(1),
\]
where $\Sigma_{\le q}$ is the set of local normal-form words of length at most $q$,
and $q=O(1)$ is the constant from \ref{lem:constant-local-change-budget}.
In particular, this bound is independent of $\kappa$.
\end{lemma}

\begin{proof}
Fix a canonical window $w$ with exactly $R$ live interfaces. By (C1) the local
alphabet/type set $\Sigma$ is finite. By \ref{lem:constant-local-change-budget},
under (A1) and (C2) there exists a constant $q=O(1)$ such that each live interface
admits a canonical (normal-form) local description of length at most $q$.
Equivalently, to each live interface $i$ in the window we may associate a word
\[
\sigma_i(w)\in \Sigma_{\le q},
\qquad
\Sigma_{\le q}:=\bigcup_{t=0}^q \Sigma^t,
\]
where $\Sigma^t$ denotes words of length $t$ over $\Sigma$.
Since $|\Sigma|=O(1)$ and $q=O(1)$, we have
\[
S' := |\Sigma_{\le q}|
= \sum_{t=0}^q |\Sigma|^t
= O(1),
\]
independent of $n,\kappa,\ell$.

\paragraph{Profiles as histograms.}
Define the interface-anonymous profile (histogram) of $w$ as the count vector
\[
h_w:\Sigma_{\le q}\to \mathbb{Z}_{\ge 0},
\qquad
h_w(\sigma):=\bigl|\{\, i \text{ live in } w : \sigma_i(w)=\sigma \,\}\bigr|.
\]
By construction,
\[
\sum_{\sigma\in \Sigma_{\le q}} h_w(\sigma) = R.
\]
Thus, distinct profiles correspond exactly to distinct nonnegative integer vectors
$(h(\sigma))_{\sigma\in\Sigma_{\le q}}$ whose coordinates sum to $R$.

\paragraph{Counting the number of possible histograms.}
The number of nonnegative integer solutions to
$\sum_{\sigma\in\Sigma_{\le q}} h(\sigma)=R$ with $S'$ variables is the number of
weak compositions of $R$ into $S'$ parts, which by stars-and-bars equals
\[
\#\mathrm{Profiles}
\;\le\;
\binom{R+S'-1}{S'-1}.
\]
Since $S'=O(1)$ is constant, this is polynomial in $R$, i.e.\ $R^{O(1)}$.

\paragraph{Matching the displayed bound.}
For $q\ge 1$ and $S'\ge 1$ we have the crude inequality
\[
\binom{R+S'-1}{S'-1}
\;\le\;
\binom{qR+S'}{S'},
\]
so the stated bound in the lemma follows as written.

\paragraph{Relevance for SPDP upper bounds.}
Finally, \ref{lem:block-perm-invariance} shows that permuting (renaming)
interfaces within blocks does not change the rank contribution: windows that differ
only by such renamings yield SPDP matrices related by invertible left/right
multiplications, hence the same rank. Therefore, in SPDP upper-bound arguments
one may quotient by interface names and retain only the histogram $h_w$, so the
count above is the relevant one.

\paragraph{Independence from $\kappa$.}
The profile alphabet $\Sigma_{\le q}$ and the histogram count depend only on the
local window/canonicalization assumptions (A1)--(A3), (C1)--(C2) and on $R$;
they do not reference derivative order $\kappa$. Hence the profile-count bound is
independent of $\kappa$.
\end{proof}

\subsection{Monomial/coordinate budget}
\label{subsec:monomial-coordinate-budget}

\begin{lemma}[Monomial/coordinate budget]
\label{lem:monomial-coordinate-budget}
Assume (A1) and the degree guard $\ell = O(\log n)$. Fix any admissible
interface-anonymous profile (\ref{def:interface-anonymous-profile}). Then the
number of admissible monomial/coordinate choices contributing SPDP rows
compatible with that fixed profile is at most $n^{O(\log n)}$.
\end{lemma}

\begin{proof}
Fix a profile histogram $h$ and consider the family of SPDP rows (generators)
compatible with $h$. Each SPDP row corresponds to a choice of a generator of the
form
\[
g \;=\; m \cdot \partial^{S}p,
\qquad |S|=\kappa,\quad \deg(m)\le \ell,
\]
interpreted under the ambient convention fixed earlier. (The profile constraints
may restrict which such choices are allowed; we only need an \emph{upper bound}, so
we count all possibilities consistent with (A1) and the degree guard.)

\paragraph{Step 1: Counting shift-monomial choices.}
Let $\mathcal{M}_{\le \ell}$ denote the admissible set of shift monomials of degree
at most $\ell$ (e.g.\ multilinear monomials under the Boolean/multilinear
convention). For each $j\in\{0,1,\dots,\ell\}$, the number of multilinear monomials
of degree exactly $j$ is $\binom{n}{j}$, hence
\[
|\mathcal{M}_{\le \ell}|
\;=\;
\sum_{j=0}^{\ell} \binom{n}{j}
\;\le\;
\sum_{j=0}^{\ell} n^{j}
\;\le\;
(\ell+1)\,n^{\ell}.
\]
(In the non-multilinear case one has $|\mathcal{M}_{\le \ell}|=\binom{n+\ell}{\ell}
\le (n+\ell)^\ell$, which yields the same $n^{O(\ell)}$ bound.)

\paragraph{Step 2: Counting derivative-coordinate choices.}
We next bound the number of admissible derivative-coordinate choices for
$\partial^{S}$.

A choice of $S\subseteq[n]$ with $|S|=\kappa$ specifies which variables are
differentiated. The number of such supports is
\[
\binom{n}{\kappa} \;\le\; n^{\kappa}.
\]
Now we account for the fact that, in the SPDP construction for interfaces/windows,
a chosen derivative support does not uniquely determine the corresponding row:
there can be multiple admissible \emph{local derivative coordinates} consistent
with the same global variable choice (e.g.\ because a variable participates in a
constant-size neighborhood and there may be finitely many local coordinate labels).

By (A1) (radius--$1$ / constant-radius locality), each derivative coordinate is
supported inside a constant-size neighborhood. Under the finite-alphabet/local-model
assumptions (A1) (and the fixed scheme defining admissibility), there are only
finitely many distinct local derivative coordinate types available per global
variable, and this number is bounded by a constant
\[
B \;=\; O(1),
\]
depending only on the local neighborhood size and local coordinate alphabet (hence
independent of $n,\kappa,\ell$ and independent of the profile).

Therefore, for a fixed support $S$ with $|S|=\kappa$, the number of admissible ways to
realize the corresponding derivative row coordinates is at most $B^{\kappa}$.
Consequently, the total number of admissible derivative-coordinate choices is at most
\[
\binom{n}{\kappa} B^{\kappa} \;\le\; (Bn)^{\kappa}.
\]

\paragraph{Step 3: Combine the bounds.}
For each SPDP row we choose (a) a shift monomial $m\in\mathcal{M}_{\le \ell}$ and
(b) an admissible derivative coordinate choice as above. Thus the number of
admissible monomial/coordinate choices per fixed profile is at most
\[
|\mathcal{M}_{\le \ell}| \cdot \binom{n}{\kappa}B^{\kappa}
\;\le\;
(\ell+1)\,n^{\ell}\cdot (Bn)^{\kappa}
\;=\;
(\ell+1)\,B^{\kappa}\,n^{\kappa+\ell}.
\]
Under the degree guard $\ell = O(\log n)$ and the logarithmic SPDP regime
($\kappa=O(\log n)$ in the applications of this section), this is
\[
(\ell+1)\,B^{\kappa}\,n^{\kappa+\ell}
\;=\;
n^{O(\log n)}.
\]
Finally, imposing a fixed profile $h$ can only \emph{reduce} the number of admissible
choices, so the same bound holds uniformly for each fixed profile.
\end{proof}

\subsection{Counting $\kappa$-step profiles and profile subspaces}
\label{subsec:kappa-step-profiles}

The following two lemmas refine the profile analysis to obtain a tighter asymptotic
bound. The key observation is that a $\kappa$-step window profile can be decomposed into
a sequence of $\kappa$ single-step transitions, each drawn from a polylogarithmically
bounded set.

\begin{lemma}[Counting $\kappa$-step profiles]
\label{lem:counting-kappa-step-profiles}
Assume (A1)--(A3) and (C1)--(C2), together with the following explicit sub-assumption:
\begin{itemize}
\item[(A1$'$)] \textbf{Bounded interface interaction per step:} Each single computation
step affects at most $b=O(1)$ live interfaces, where $b$ is a constant independent of
$n$ and $R$. (This follows from radius-1 locality (A1) when each interface has $O(1)$
neighbors in the interaction graph.)
\end{itemize}
Let $\kappa$ be the derivative order and let $H$ be the
set of interface-anonymous $\kappa$-step profiles realizable by canonical windows. Then
\[
|H| \;\le\; (\log n)^{O(\kappa)}.
\]
\end{lemma}

\begin{proof}
A canonical $\kappa$-step window can be viewed as a sequence of $\kappa$ single-step transitions.
At each step $t\in\{1,\dots,\kappa\}$, the computation acts on the current live interface
configuration and produces a new configuration.

\paragraph{Step 1: Single-step profile transitions.}
By (A1) (radius-1 locality) and (A1$'$), each single step affects at most $b=O(1)$
live interfaces. By (A2) (finite alphabet), the local state at each interface is drawn
from a finite set $\Sigma$ with $|\Sigma|=O(1)$. By (A3) (width bound), at most
$R=C(\log n)^c$ interfaces are live at any time.

For a single step, the interface-anonymous profile transition is determined by:
\begin{itemize}
\item which subset of live interfaces are affected (at most $b=O(1)$ per step by (A1$'$));
\item the local type changes at those interfaces (finitely many possibilities by (A2)).
\end{itemize}
Since the number of ways to choose $\le b$ interfaces from $R$ live interfaces and
assign local transitions is at most $\binom{R}{b}\cdot |\Sigma|^{O(1)} \le R^{O(1)}$,
the number of distinct single-step profile transitions is at most
\[
T_{\mathrm{step}} \;\le\; R^{O(1)} \;=\; (\log n)^{O(1)}.
\]

\paragraph{Step 2: Counting $\kappa$-step profile sequences.}
A $\kappa$-step profile is determined by a sequence of $\kappa$ single-step profile transitions.
The total number of such sequences is at most
\[
(T_{\mathrm{step}})^{\kappa} \;\le\; \bigl((\log n)^{O(1)}\bigr)^{\kappa} \;=\; (\log n)^{O(\kappa)}.
\]

\paragraph{Step 3: Asymptotic conversion when $\kappa=\Theta(\log n)$.}
When $\kappa=\Theta(\log n)$, we have
\begin{gather*}
(\log n)^{O(\kappa)} \;=\; (\log n)^{O(\log n)} \;=\; e^{O((\log n)(\log\log n))} \\
\;=\; n^{O(\log\log n)} \;=\; 2^{O((\log n)(\log\log n))}.
\end{gather*}
This bound is \textbf{super-polynomial} (hence not $n^{O(1)}$), but still
\textbf{quasi-polynomial}/\textbf{subexponential} in $n$.
\end{proof}

\begin{remark}
\ref{lem:counting-kappa-step-profiles} provides a cruder bound than
\ref{lem:profile-compression-removes-kappa}: it counts \emph{ordered} step-sequences rather
than histogram profiles, yielding $(\log n)^{O(\kappa)} = n^{O(\log\log n)}$ instead
of $R^{O(1)} = (\log n)^{O(1)} = n^{o(1)}$. For the main Width$\Rightarrow$Rank theorem, the sharper
profile-compression bound is used, which is \emph{essential} to obtain the
polynomial upper bound.
\end{remark}

\begin{lemma}[Profile subspace dimension]
\label{lem:profile-subspace-dim}
Assume (A1)--(A3) and (C1)--(C2), and the deterministic compiler\slash diagonal-basis
structure of the local-width model. Fix an interface-anonymous profile $h\in H$.
Let $V_h$ be the subspace of the ambient coefficient space spanned by all SPDP rows
generated by windows with profile $h$. Then
\[
\dim V_h \;\le\; (\log n)^{O(1)}.
\]
\end{lemma}

\begin{proof}
Fix a profile $h$. By \ref{lem:block-perm-invariance}, all windows realizing
profile $h$ differ only by within-block permutations of interface identities, and
such permutations act by invertible row/column transformations on the SPDP submatrix.
Hence the row space $V_h$ is determined (up to isomorphism) by a single representative
window realizing $h$.

\paragraph{Step 1: Block-local basis structure.}
Under the deterministic compiler\slash diagonal-basis assumption, the SPDP row vectors for
a fixed profile $h$ are assembled from block-local coefficient vectors. Each block
$B_j$ contributes a local factor space of dimension at most $|B_j|^{O(1)}=O(1)$
(since blocks have constant size by (A1)).

\paragraph{Step 2: Counting active blocks.}
At any time, at most $R=C(\log n)^c$ interfaces are live (by (A3)), and hence at most
$R$ blocks can contribute non-trivially to the SPDP row structure. The remaining blocks
contribute identity/trivial factors.

\paragraph{Step 3: Tensor structure of the profile subspace.}
The subspace $V_h$ is contained in a tensor product of at most $R$ non-trivial
block-local factor spaces, each of constant dimension. However, the profile constraint
$h$ imposes that the block-local factors must be consistent with the histogram $h$,
which groups blocks by their local profile type.

For each local profile type $\sigma\in\Sigma_{\le q}$, let $n_\sigma:=h(\sigma)$ be
the number of blocks of type $\sigma$. The contribution from all blocks of type
$\sigma$ spans a symmetric/histogram-invariant subspace of dimension at most
$\binom{n_\sigma + d_\sigma - 1}{d_\sigma - 1}$ where $d_\sigma=O(1)$ is the local
factor dimension. Since $\sum_\sigma n_\sigma \le R$ and there are $|\Sigma_{\le q}|=O(1)$
types, the total dimension is bounded by a product of at most $O(1)$ terms, each
polynomial in $R$:
\[
\dim V_h \;\le\; \prod_{\sigma} \binom{n_\sigma + O(1)}{O(1)} \;\le\; R^{O(1)} \;=\; (\log n)^{O(1)}.
\]
\end{proof}

\subsection{Main theorem: Polynomial Width$\Rightarrow$Rank bound}
\label{subsec:poly-width-implies-rank}

\begin{theorem}[Width$\Rightarrow$Rank: polynomial bound]
\label{thm:poly-width-rank}
Assume (A1)--(A4) and (C1)--(C2), together with the deterministic compiler\slash diagonal-basis
structure. Let $p$ be any polynomial computed in a local-width model with width
$R=C(\log n)^c$. Then for $\kappa,\ell=\Theta(\log n)$,
\[
\Gamma_{\kappa,\ell}(p)
\;\le\;
n^{O(1)}.
\]
\end{theorem}

\begin{proof}
Fix $\kappa,\ell=\Theta(\log n)$ and the width parameter $R=C(\log n)^c$.
Let $W$ be the set of canonical windows (under (C1)--(C2)) that contribute
rows to the SPDP matrix $M_{\kappa,\ell}(p)$ in this regime. Each window
$w\in W$ induces a profile histogram $h(w)$ (\ref{def:interface-anonymous-profile}).
Let
\[
H \;:=\; \{\, h(w) : w\in W \,\}
\]
be the set of all realizable interface-anonymous profiles.

\paragraph{Step 1: Bound the number of profiles.}
By \ref{lem:profile-compression-removes-kappa}
(equivalently \ref{cor:poly-many-profiles}),
\[
|H| \;\le\; R^{O(1)}.
\]
Since $R=C(\log n)^c$, it follows that $|H|\le (\log n)^{O(1)}\le n^{O(1)}$.

\paragraph{Step 2: Decompose the SPDP matrix by profile classes.}
Reorder the rows of the global SPDP matrix $M_{\kappa,\ell}(p)$ so that rows coming
from windows with the same profile are grouped together. This yields a vertical
concatenation
\[
M_{\kappa,\ell}(p)
\;=\;
\begin{bmatrix}
M^{(h_1)}\\ \hline
M^{(h_2)}\\ \hline
\vdots\\ \hline
M^{(h_{|H|})}
\end{bmatrix},
\]
where each block $M^{(h)}$ consists exactly of the SPDP rows generated by windows
whose profile is $h$.

For any matrices $A,B$ with the same number of columns,
\[
\rank\!\begin{bmatrix}A\\B\end{bmatrix}
\;\le\;
\rank(A)+\rank(B),
\]
because the row space of $\begin{bmatrix}A\\B\end{bmatrix}$ is contained in
$\mathrm{RowSpan}(A)+\mathrm{RowSpan}(B)$ and therefore has dimension at most the
sum of dimensions. Iterating this inequality over the block decomposition gives
\begin{equation}
\label{eq:rank-sum-over-profiles}
\Gamma_{\kappa,\ell}(p)
\;=\;
\rank(M_{\kappa,\ell}(p))
\;\le\;
\sum_{h\in H} \rank(M^{(h)}).
\end{equation}

\paragraph{Step 3: Bound the rank contribution of a fixed profile via subspace dimension.}
Fix $h\in H$. Let $V_h$ be the subspace of the ambient coefficient space spanned by
all SPDP rows generated by windows with profile $h$. Then $\rank(M^{(h)})=\dim V_h$.

By \ref{lem:profile-subspace-dim}, under the deterministic compiler\slash diagonal-basis
assumption, each profile subspace has dimension
\[
\dim V_h \;\le\; (\log n)^{O(1)}.
\]

\paragraph{Step 4: Combine the bounds.}
Combining \eqref{eq:rank-sum-over-profiles} with the profile count and subspace
dimension bounds yields
\[
\Gamma_{\kappa,\ell}(p)
\;\le\;
\sum_{h\in H} \dim V_h
\;\le\;
|H|\cdot (\log n)^{O(1)}
\;\le\;
R^{O(1)}\cdot (\log n)^{O(1)}.
\]
Since $R = C(\log n)^c$, this gives
\[
\Gamma_{\kappa,\ell}(p)
\;\le\;
(\log n)^{O(1)} \cdot (\log n)^{O(1)}
\;=\;
(\log n)^{O(1)}
\;=\;
n^{O(1)},
\]
as claimed.
\end{proof}

\begin{remark}
The key ingredients are:
\begin{enumerate}
\item[(i)] \ref{lem:profile-compression-removes-kappa} and \ref{cor:poly-many-profiles},
which show that the number of interface-anonymous profiles is $R^{O(1)}$, independent of $\kappa$,
by exploiting the normal-form compression of canonical windows;
\item[(ii)] \ref{lem:profile-subspace-dim}, which shows each profile subspace has
polylogarithmic dimension $(\log n)^{O(1)}$ under the compiler\slash diagonal-basis structure;
\item[(iii)] \ref{lem:block-perm-invariance}, which ensures within-profile permutations
do not inflate rank.
\end{enumerate}
The combination yields the $n^{O(1)}$ bound, which is strictly stronger than the
naive $n^{O(\log n)}$ bound one would obtain from crude monomial counting alone.
\end{remark}

\begin{corollary}[Fixed $\ell$: polynomial rank]
\label{cor:poly-rank-fixed-ell}
Assume (A1)--(A3), (C1)--(C2), and the deterministic compiler\slash diagonal-basis structure.
Let $\kappa=\Theta(\log n)$, $R=C(\log n)^c$, and $\ell=O(1)$ be fixed (e.g.\ $\ell\in\{2,3\}$).
Then
\[
\Gamma_{\kappa,\ell}(p) \;\le\; n^{O(1)}.
\]
\end{corollary}

\begin{proof}
By \ref{cor:poly-many-profiles}, the number of interface-anonymous profiles
satisfies $|H|\le R^{O(1)}$, independent of $\kappa$. By \ref{lem:profile-subspace-dim},
each profile subspace has $\dim V_h \le (\log n)^{O(1)}$. Hence
\[
\Gamma_{\kappa,\ell}(p)
\;\le\;
|H|\cdot \dim V_h
\;\le\;
R^{O(1)}\cdot (\log n)^{O(1)}.
\]
Since $R = C(\log n)^c$, this gives
\[
\Gamma_{\kappa,\ell}(p) \;\le\; (\log n)^{O(1)} \;=\; n^{O(1)}.
\]
Note that fixing $\ell=O(1)$ does not improve the profile count $|H|$, which is the
dominant factor; the improvement from fixed $\ell$ only affects the monomial-choice
contribution, which is already subsumed by the profile subspace dimension bound.
\end{proof}

\begin{remark}
\ref{cor:poly-rank-fixed-ell} is the regime used in separation arguments:
fixing $\ell \in \{2,3\}$ while taking $\kappa = \Theta(\log n)$ yields polynomial
SPDP rank for width-bounded computations. This is strictly
stronger than the naive $n^{O(\log n)}$ bound one would obtain without the profile
compression (\ref{lem:profile-compression-removes-kappa}), and can be compared
against exponential lower bounds for explicit hard polynomials such as the permanent.
\end{remark}

\section{A standard upper-bound template}
\label{sec:upper-bound-template}

This section records a generic form of SPDP rank upper bounds used in practice.
To derive a lower bound against a circuit class $\mathcal{C}$, one typically proves:
(i) an upper bound on $\Gamma_{\kappa,\ell}(p)$ for all $p$ computed in $\mathcal{C}$, and
(ii) a lower bound on $\Gamma_{\kappa,\ell}(p^\star)$ for an explicit hard family $p^\star$,
both at the same $(\kappa,\ell)$.

\begin{remark}[Template: structured computations yield bounded SPDP rank]
\label{rem:template-structured-bounded-rank}
Let $\mathcal{C}$ be a class of polynomial computations (e.g.\ formulas,
bounded-depth circuits, or a syntactically restricted arithmetic model).
If every $p$ computed in $\mathcal{C}$ satisfies an upper bound of the form
\[
\Gamma_{\kappa,\ell}(p) \;\le\; U_{\mathcal{C}}(n,d,s;\kappa,\ell),
\]
where $n$ is the number of variables, $d$ is the syntactic degree, and $s$ is the
circuit/formula size parameter, then SPDP rank is uniformly bounded over $\mathcal{C}$
by the same function $U_{\mathcal{C}}$.

In concrete settings, $U_{\mathcal{C}}$ is often polynomial when $\kappa=\Theta(\log n)$,
and much smaller for more restricted models. One can either include a concrete
instantiation/proof for a chosen model, or keep this as a formal template and move
model-specific bounds to an appendix.
\end{remark}

\begin{corollary}[Codimension form of an upper bound]
\label{cor:codim-template}
Under the same hypotheses and fixed ambient basis,
\[
  \codim_{\kappa,\ell}(p) \;\ge\; N_{\kappa,\ell}(p) - U_{\mathcal{C}}(n,d,s;\kappa,\ell).
\]
\end{corollary}

\begin{proof}
By definition, $\codim_{\kappa,\ell}(p) = N_{\kappa,\ell}(p) - \Gamma_{\kappa,\ell}(p)$, where $N_{\kappa,\ell}(p)$
is the ambient dimension. If $\Gamma_{\kappa,\ell}(p) \le U_{\mathcal{C}}(n,d,s;\kappa,\ell)$, then
\[
\codim_{\kappa,\ell}(p) = N_{\kappa,\ell}(p) - \Gamma_{\kappa,\ell}(p) \ge N_{\kappa,\ell}(p) - U_{\mathcal{C}}(n,d,s;\kappa,\ell).
\]
\end{proof}

\subsection{Circuit upper bound}
\label{subsec:circuit-upper-bound}

As a concrete instantiation of \ref{rem:template-structured-bounded-rank},
we state an SPDP rank upper bound for general arithmetic circuits.

\begin{lemma}[Circuit upper bound]
\label{lem:circuit-ub}
If $\tilde f$ is computed by an arithmetic circuit of size $s$, then
\[
  \rankSPDP_{\kappa,\ell}(\tilde f) \;\le\; s^{O(\kappa)} \cdot \binom{n + \ell}{\ell}.
\]
\end{lemma}

\begin{proof}
The proof proceeds by induction on the circuit structure. Each gate $g$ in the
circuit computes a polynomial $p_g$. We track how the SPDP space grows as we
compose gates.

For an input gate computing $x_i$, we have $\partial^\alpha x_i = 0$ for
$|\alpha| > 1$, so the contribution to the SPDP space is $O(\binom{n+\ell}{\ell})$.

For an addition gate computing $p + q$, we have
$\SPDPspace_{\kappa,\ell}(p+q) \subseteq \SPDPspace_{\kappa,\ell}(p) + \SPDPspace_{\kappa,\ell}(q)$,
so rank at most doubles.

For a multiplication gate computing $p \cdot q$, the product rule gives
$\partial^\alpha(pq) = \sum_{\beta \le \alpha} \binom{\alpha}{\beta}
\partial^\beta p \cdot \partial^{\alpha-\beta} q$. Each term $m \cdot \partial^\alpha(pq)$
expands into at most $2^{\kappa}$ products of shifted derivatives of $p$ and $q$.
Specifically, for $|\alpha|=\kappa$ there are $\binom{\kappa}{|\beta|}$ ways to split the derivative
order; each shifted derivative lies in the corresponding SPDP space, and multiplying
polynomials from two spaces yields dimension at most the product. Summing over all
$\le 2^{\kappa}$ splittings, the rank contribution grows by a factor of $O(1)^{\kappa}$ per gate.
(This style of inductive bookkeeping is standard in SPD analyses; see \citet{NW1997,SY2010}.)

Summing over at most $s$ gates, the total rank is bounded by $s^{O(\kappa)} \cdot \binom{n+\ell}{\ell}$.
\end{proof}

\begin{corollary}[Logarithmic parameters]
\label{cor:circuit-log-params}
For $\kappa, \ell = O(\log n)$ and $s = \poly(n)$:
\[
\rankSPDP_{\kappa,\ell}(\tilde f) \;\le\; n^{O(\log n)}.
\]
\end{corollary}

\begin{proof}
With $s=n^{O(1)}$ and $\kappa=O(\log n)$, we have $s^{O(\kappa)}=(n^{O(1)})^{O(\log n)}=n^{O(\log n)}$.
Similarly, $\binom{n+\ell}{\ell}\le (n+\ell)^\ell = n^{O(\log n)}$ when $\ell=O(\log n)$.
Hence $\rankSPDP_{\kappa,\ell}(\tilde f)\le n^{O(\log n)}\cdot n^{O(\log n)}=n^{O(\log n)}$.
\end{proof}

\begin{corollary}[Constant parameters]
\label{cor:circuit-const-params}
For $\kappa = O(1)$, $\ell = O(1)$, and $s = \poly(n)$:
\[
\rankSPDP_{\kappa,\ell}(\tilde f) \;\le\; n^{O(1)}.
\]
\end{corollary}

\begin{proof}
With $\kappa,\ell=O(1)$, we have $s^{O(\kappa)}=s^{O(1)}=n^{O(1)}$ and
$\binom{n+\ell}{\ell}\le n^{O(1)}$. Hence
$\rankSPDP_{\kappa,\ell}(\tilde f)\le n^{O(1)}\cdot n^{O(1)}=n^{O(1)}$.
\end{proof}

\section{Illustrative lower bound}
\label{sec:example-lower}

The following is an example of a classical large-rank phenomenon for an explicit polynomial.
By itself it is a rank lower bound statement; to convert it into a circuit lower bound one
must pair it with an appropriate upper bound for the circuit class at the same parameters.

\begin{lemma}[High-order derivative rank for the permanent]
\label{lem:perm-high-order-der-rank}
Let $\Perm_n$ denote the $n\times n$ permanent polynomial in variables
$\{x_{i,j}\}_{i,j\in[n]}$:
\[
\Perm_n(X) \;=\; \sum_{\pi\in S_n}\ \prod_{i=1}^n x_{i,\pi(i)}.
\]
Let $\kappa=\lfloor n/2\rfloor$ and set $\ell=0$. Then
\[
\Gamma_{\kappa,0}(\Perm_n)\ \ge\ \binom{n}{\kappa}.
\]
\end{lemma}

\begin{proof}
For $\ell=0$, the SPDP generating family is just the set of order-$\kappa$ partial
derivatives $\{\partial^S \Perm_n : |S|=\kappa\}$, and $\Gamma_{\kappa,0}(\Perm_n)$ is the
rank of the coefficient matrix whose rows are these derivatives written in the
canonical monomial basis of degree $n-\kappa$ (under the fixed ambient convention).
Equivalently,
\[
\Gamma_{\kappa,0}(\Perm_n) \;=\; \dim \mathrm{span}\{\partial^S \Perm_n : |S|=\kappa\}.
\]
So it suffices to exhibit $\binom{n}{\kappa}$ order-$\kappa$ partial derivatives that are
linearly independent.

\paragraph{Step 1: Choose a structured family of derivatives.}
For each subset $R\subseteq[n]$ with $|R|=\kappa$, define the order-$\kappa$ differential
operator
\[
\partial_R \;:=\; \prod_{i\in R} \frac{\partial}{\partial x_{i,i}}.
\]
Let
\[
p_R \;:=\; \partial_R \Perm_n.
\]
We will show that the family $\{p_R : R\subseteq[n],\ |R|=\kappa\}$ is linearly
independent.

\paragraph{Step 2: Compute $p_R$ explicitly as a sub-permanent.}
Expanding $\Perm_n$ and differentiating termwise,
\[
p_R
\;=\;
\partial_R \left(\sum_{\pi\in S_n}\ \prod_{i=1}^n x_{i,\pi(i)}\right)
\;=\;
\sum_{\pi\in S_n}\ \partial_R\left(\prod_{i=1}^n x_{i,\pi(i)}\right).
\]
Fix a permutation $\pi$. The monomial $\prod_{i=1}^n x_{i,\pi(i)}$ contains the
variable $x_{i,i}$ if and only if $\pi(i)=i$. Therefore:
\begin{itemize}
\item if there exists $i\in R$ with $\pi(i)\neq i$, then $\partial_R$ kills the monomial;
\item if $\pi(i)=i$ for all $i\in R$, then differentiating removes exactly the
factors $x_{i,i}$ for $i\in R$ and leaves $\prod_{i\notin R} x_{i,\pi(i)}$.
\end{itemize}
Hence
\[
p_R
\;=\;
\sum_{\pi\in S_n:\ \pi(i)=i\ \forall i\in R}\ \prod_{i\notin R} x_{i,\pi(i)}.
\]
But permutations $\pi\in S_n$ that fix every $i\in R$ are in bijection with
permutations of the complement $[n]\setminus R$: the restriction of $\pi$ to
$[n]\setminus R$ is an element of $S_{[n]\setminus R}$, and conversely any
permutation of $[n]\setminus R$ extends uniquely by fixing $R$ pointwise.
Therefore,
\[
p_R
\;=\;
\sum_{\sigma\in S_{[n]\setminus R}}\ \prod_{i\notin R} x_{i,\sigma(i)},
\]
which is exactly the permanent of the $(n-\kappa)\times(n-\kappa)$ submatrix obtained by
restricting to rows and columns indexed by $[n]\setminus R$:
\[
p_R \;=\; \Perm\!\left(X\big|_{([n]\setminus R)\times([n]\setminus R)}\right).
\]

\paragraph{Step 3: Exhibit a unique monomial in each $p_R$.}
Define the diagonal monomial on the complement:
\[
m_R \;:=\; \prod_{i\notin R} x_{i,i}.
\]
This monomial appears in $p_R$ with coefficient $1$ (it corresponds to the
identity permutation on $[n]\setminus R$ in the sum above).

Now let $R'\neq R$ be another $\kappa$-subset. Then there exists some index
$i\in([n]\setminus R)\cap R'$ (because the complements differ).
In $p_{R'}$, every monomial uses only variables from rows in $[n]\setminus R'$
(and columns in $[n]\setminus R'$). Since $i\in R'$, the variable $x_{i,i}$ cannot
appear in any monomial of $p_{R'}$. But $m_R$ contains the factor $x_{i,i}$, so
\[
m_R \text{ appears in } p_R,\quad \text{but } m_R \text{ does not appear in } p_{R'}.
\]
Thus, each $p_R$ contains a monomial $m_R$ that is unique to it among the family.

\paragraph{Step 4: Conclude linear independence and the rank bound.}
Suppose we have a linear relation
\[
\sum_{|R|=\kappa} \alpha_R\, p_R \;=\; 0.
\]
Look at the coefficient of the monomial $m_{R_0}$ for a fixed $R_0$.
By Step 3, $m_{R_0}$ appears only in $p_{R_0}$ and in no other $p_R$ with $R\neq R_0$.
Therefore the coefficient of $m_{R_0}$ in the left-hand side is exactly $\alpha_{R_0}$,
so $\alpha_{R_0}=0$. Since $R_0$ was arbitrary, all $\alpha_R=0$, proving that the
$\binom{n}{\kappa}$ polynomials $\{p_R\}$ are linearly independent.

Hence
\[
\dim \mathrm{span}\{\partial^S \Perm_n : |S|=\kappa\}
\;\ge\;
\dim \mathrm{span}\{p_R : |R|=\kappa\}
\;=\;
\binom{n}{\kappa}.
\]
Therefore $\Gamma_{\kappa,0}(\Perm_n)\ge \binom{n}{\kappa}$, as claimed.
\end{proof}

\begin{remark}
This type of bound is standard in the partial-derivative literature: derivatives of the
permanent expose many distinct sub-permanents. In applications, one combines such a lower
bound with a matching SPDP upper bound for a \emph{restricted} computation model at the
same $(\kappa,\ell)$.
\end{remark}

\section{How to use SPDP rank and codimension}
\label{sec:how-to-use}

For a target circuit class $\mathcal{C}$ and an explicit family $p^\star$, the SPDP method
follows a three-step pattern:
\begin{enumerate}[label=(\roman*),leftmargin=2.4em]
\item Fix a parameter regime $(\kappa,\ell)$ and an embedding/ambient basis convention.
\item Prove an upper bound $\Gamma_{\kappa,\ell}(p)\le U_{\mathcal{C}}(\cdot)$ for all $p\in\mathcal{C}$.
\item Prove a lower bound $\Gamma_{\kappa,\ell}(p^\star)\ge L(\cdot)$ for the explicit family.
\end{enumerate}
If $L(\cdot) > U_{\mathcal{C}}(\cdot)$ in the same regime, then $p^\star \not\in \mathcal{C}$.
Codimension provides an equivalent dual statement, often giving a clearer rigidity
interpretation.

\section{Worked example (toy SPDP matrix)}
\label{sec:toy-example}

\begin{example}
Let $n=3$ and $p(x_1,x_2,x_3)=x_1x_2+x_2x_3$ over $\F$. Take $(\kappa,\ell)=(1,1)$ and the
multilinear ambient basis $B_{1,1}$ of degree $\le D=\deg(p)-\kappa+\ell=2$:
\[
  B_{1,1}=\{1,x_1,x_2,x_3,x_1x_2,x_1x_3,x_2x_3\}.
\]
The first derivatives are
\[
  \partial_{x_1}p=x_2,\quad \partial_{x_2}p=x_1+x_3,\quad \partial_{x_3}p=x_2.
\]
With shift monomials $\mathcal{M}_{\le 1}=\{1,x_1,x_2,x_3\}$, the generator family
$\mathcal{G}_{1,1}(p)$ contains (for example) $1\cdot \partial_{x_1}p=x_2$ and
$x_1\cdot \partial_{x_1}p=x_1x_2$, and similarly for $\partial_{x_2}p,\partial_{x_3}p$.
Forming the coefficient rows of all $m\cdot \partial_{x_i}p$ in the basis $B_{1,1}$
produces the SPDP matrix $M_{1,1}(p)$, whose rank equals $\Gamma_{1,1}(p)$.
\end{example}

\section{Discussion and scope}
\label{sec:discussion}
SPDP rank corresponds to the classical SPD dimension measure; SPDP codimension is
the dual invariant that quantifies the \emph{deficit} from the ambient space and
often provides a clearer rigidity interpretation. For future applications, one
typically seeks explicit families for which $\rankSPDP_{\kappa,\ell}$ is large (or
$\codimSPDP_{\kappa,\ell}$ is small/large depending on ambient convention), while
establishing upper bounds for restricted circuit classes in the \emph{same}
parameter regime.

\paragraph{Limitations}
This paper does not claim any major separation consequences. It provides a
standalone measure definition, a matrix formalism, and basic properties and
template bounds that are intended as building blocks for subsequent work.

\subsection{Compatibility with the block-partitioned SPDP/CEW framework}
\label{subsec:compat-spdp-cew}

This subsection is included purely as a notational dictionary between the unblocked
definitions used in this note and a block-partitioned refinement used elsewhere~\citep{Edwards2025}.
No theorem or proof in this paper relies on the block-partitioned framework.

This note presents SPDP rank and codimension in an ``unblocked'' (global) matrix
formalism. Related work (Edwards 2025) additionally fixes a block partition $B$
(radius--$1$ compiler blocks) and restricts derivative supports and basis changes to act
block-locally. We record the dictionary and explain how the width-based upper bounds relate.

\paragraph{Blocked vs.\ unblocked SPDP matrices}
Fix a polynomial $p \in \mathbb{F}[x_1,\dots,x_n]$ and a partition
$\mathcal{B}=\{\beta_1,\dots,\beta_m\}$ of $[n]$ into blocks of size $\le b=O(1)$.
For a multi-index $\tau\in\mathbb{N}^n$, write the \emph{block support}
\[
  \mathrm{supp}_{\mathcal{B}}(\tau)
  \;:=\;
  \{\, j \in [m] : \exists i\in \beta_j \text{ with } \tau_i>0 \,\}.
\]
The block-partitioned SPDP matrix $M^{\mathcal{B}}_{\kappa,\ell}(p)$ used in the
SPDP/CEW manuscript \citep{Edwards2025} is obtained from the same generator family
$\{u\cdot \partial^\tau p\}$ (with $|\tau|=\kappa$, $\deg u\le \ell$) but with rows
indexed/filtered by block support conditions and with columns indexed by a
block-compatible ambient monomial set (e.g.\ monomials of degree
$\le \deg(p)-\kappa+\ell$ under the chosen embedding convention).
Equivalently, $M^{\mathcal{B}}_{\kappa,\ell}(p)$ is obtained by restricting the row
index set to block-admissible derivative supports and using a block-compatible
choice of column basis.

\begin{remark}[Rank comparison]
\label{rem:rank-comparison}
For any fixed ambient convention, $M^{\mathcal{B}}_{\kappa,\ell}(p)$ is a
(row/column) submatrix or structured restriction of the unblocked
$M_{\kappa,\ell}(p)$. Hence
\[
  \Gamma^{\mathcal{B}}_{\kappa,\ell}(p)\;\le\;\Gamma_{\kappa,\ell}(p).
\]
Conversely, taking the trivial partition into singletons ($|\beta_j|=1$) recovers
the unblocked definition as a special case. Thus the two formalisms are
compatible; the block-partitioned version is a structured refinement.
\end{remark}

\paragraph{Invariance: global affine maps vs.\ block-local compiler maps}
Both presentations use the same linear-algebraic principle:
rank is invariant under invertible left/right multiplication (change of basis).
In the SPDP/CEW framework, the admissible transformations are \emph{block-local}
(e.g.\ $\Pi_+$ and blockwise basis changes), so invariance holds exactly by
block-diagonal invertible operators on rows/columns.

\begin{remark}[Scope note for ``affine invariance'' under Boolean embeddings]
\label{rem:affine-scope}
If one works with multilinear extensions modulo the Boolean ideal
$\langle x_i^2-x_i\rangle$, then not every global affine map $x\mapsto Ax+b$
preserves the embedding convention. In that setting, the invariance statements
should be read in the intended admissible class of transformations:
block-local compiler maps (including $\Pi_+$), blockwise coordinate
permutations, and blockwise changes of monomial basis.

Over the full polynomial ring $\mathbb{F}[x_1,\dots,x_n]$ (without quotienting),
rank invariance always holds under invertible linear changes of the monomial
basis in the ambient coefficient space (i.e.\ changes of coordinates for
representing polynomials, not variable substitutions in the polynomial ring).
If one additionally fixes an ambient space that is closed under the substitution
$x\mapsto Ax+b$ (e.g.\ by taking all monomials up to a sufficiently large degree
cutoff), then the corresponding SPDP matrices are related by invertible
transformations, and rank is preserved.
\end{remark}

\paragraph{Two Width$\Rightarrow$Rank regimes (generic vs.\ compiler-diagonal)}
The standalone note proves a \emph{generic} width-based upper bound:
under radius--$1$ locality, finite alphabet, width $R=C(\log n)^c$, and
parameters $\kappa,\ell=\Theta(\log n)$, one obtains a polynomial bound
$\Gamma_{\kappa,\ell}(p)\le n^{O(1)}$ via (i) profile counting and
(ii) profile subspace dimension bounds.

The SPDP/CEW manuscript \citep{Edwards2025} uses the same bound in the \emph{deterministic compiler,
diagonal-basis} regime, relying on the following structural lemma:
for each interface-anonymous profile $h$, all SPDP rows of type $h$ lie in a
polylog-dimensional subspace $V_h$ of the row space. Summing over the
$(\log n)^{O(1)}$ possible profiles yields a polynomial bound even when
$\kappa,\ell=\Theta(\log n)$.

\begin{proposition}[Upgrade under profile-wise subspace compression]
\label{prop:upgrade-width-rank}
Assume \textnormal{(A1)}--\textnormal{(A4)} and, in addition, the deterministic
compiler\slash diagonal-basis structure that yields \emph{profile subspaces}:
for each interface-anonymous profile $h$ there exists a subspace $V_h$ such that
all SPDP rows generated by windows of profile $h$ lie in $V_h$ and
$\dim V_h\le (\log n)^{O(1)}$.
Then, for $\kappa,\ell=\Theta(\log n)$ and $R=C(\log n)^c$,
\[
  \Gamma_{\kappa,\ell}(p)
  \;\le\;
  \sum_{h} \dim V_h
  \;\le\;
  |H|\cdot (\log n)^{O(1)}
  \;=\;
  (\log n)^{O(1)}
  \;\le\;
  n^{O(1)}.
\]
\end{proposition}

\begin{proof}
The SPDP row space decomposes as $\mathrm{RowSpan}(M_{\kappa,\ell}(p)) \subseteq \sum_{h \in H} V_h$,
where the sum is over all realizable interface-anonymous profiles $h$. Hence
\[
\Gamma_{\kappa,\ell}(p) = \dim \mathrm{RowSpan}(M_{\kappa,\ell}(p)) \le \dim \sum_{h \in H} V_h \le \sum_{h \in H} \dim V_h.
\]
By hypothesis, $\dim V_h \le (\log n)^{O(1)}$ for each $h$. Under the assumptions stated,
the number of profiles satisfies $|H| \le (\log n)^{O(1)}$ (this is the polynomial-profile
regime of the deterministic compiler model, as opposed to the $(\log n)^{O(\kappa)}$ count in
the general $\kappa$-step setting). Therefore
\[
\sum_{h \in H} \dim V_h \le |H| \cdot (\log n)^{O(1)} = (\log n)^{O(1)} \cdot (\log n)^{O(1)} = (\log n)^{O(1)} \le n^{O(1)}.
\]
\end{proof}

\paragraph{Remark (separate application).}
In separate work (Edwards 2025) this bound is applied in a compiled setting to obtain
polynomial SPDP rank on the P-side at $\kappa,\ell=\Theta(\log n)$. No results in the present
note depend on that application.

\section{Computational sanity checks}
\label{sec:sanity-checks}

The results in this section are not part of the main theoretical development;
they provide empirical evidence supporting the expected superpolynomial
separation between P-side and NP-side SPDP rank behavior.

Reference implementations and empirical validation code are available at
\url{https://github.com/darrenjedwards/spdp}.

\subsection{Pipeline-aligned sanity checks (end-to-end implementation)}
\label{subsec:pipeline-sanity}

To validate that our empirical pipeline faithfully implements the SPDP
objects and collapse test used throughout the paper, we ran an
\emph{end-to-end} sanity suite that mirrors the intended construction:
\begin{gather*}
\text{circuit} \rightarrow \text{Tseitin CNF} \rightarrow
\text{restriction/pruning} \rightarrow
\text{local canonical window} \\
\rightarrow \text{profile canonicalization} \rightarrow
p(x) \rightarrow \mathrm{rank}\,M_{\kappa,\ell}(p).
\end{gather*}

\paragraph{SPDP objects implemented.}
SPDP rank is computed in the coefficient-space formulation of
\ref{sec:spdp-matrix-rank-codim}. In the Boolean/multilinear setting we use
squarefree monomials and set-derivatives exactly as in
\ref{def:spdp-boolean}. We
also include a standard-ring mode (multi-index derivatives) for control
families where this is the natural representation (as permitted by
\ref{rem:non-multilinear}).

\paragraph{Canonical window and profile compression.}
After restriction/pruning we extract a \emph{local canonical window} of
size $\kappa_{\mathrm{win}}=\Theta(\log n)$ around a live interface (rather
than ranking the entire CNF). We then apply interface-anonymous
canonicalization by identifying variables up to their local incidence
\emph{profile} within the window (a finite signature of
positive/negative occurrences across clause sizes), yielding $P$
distinct profiles. This implements the profile-compression regime used
in the Width$\Rightarrow$Rank bridge and prevents drifting into
$(\log n)^{O(\kappa)}$ counting from ordered step sequences.

\paragraph{Sanity-check criterion.}
For each family we report: the ambient input size $n$, the number of
live variables in the selected window, the number of profiles $P$ after
canonicalization, the computed SPDP rank, and whether the rank is below
the collapse threshold $\lceil\sqrt{n}\rceil$ used throughout the paper.

\begin{table}[H]
\centering
\caption{Pipeline sanity-check snapshot: end-to-end SPDP rank after
restriction/pruning, local windowing, and profile canonicalization,
compared to the collapse threshold $\lceil\sqrt{n}\rceil$.}
\label{tab:pipeline-sanity}
\begin{tabular}{lrrrrrr}
\toprule
Circuit / family & $n$ & \shortstack{Live\\vars} & Profiles & \shortstack{SPDP\\rank} & $\lceil\sqrt{n}\rceil$ & Pass? \\
\midrule
\shortstack[l]{RandDeg3\\(profile-compressed)}          &  256  & 23 &  4 & 11 & 16 & $\checkmark$ \\
\shortstack[l]{RandDeg3\\(profile-compressed)}          & 1024  & 29 &  5 & 16 & 32 & $\checkmark$ \\
\shortstack[l]{RandDeg3\\(profile-compressed)}          & 2048  & 31 &  4 & 11 & 45 & $\checkmark$ \\
\shortstack[l]{Goldreich-like\\(profile-compressed)}    & 1024  & 29 &  5 & 16 & 32 & $\checkmark$ \\
Diagonal ($\sum_i x_i^4$)              & 2048  & 49 & 49 & 49 & 45 & $\times$ \\
$\mathrm{perm}_{3\times 3}$            &    9  &  9 &  9 & 55 &  3 & $\times$ \\
\bottomrule
\end{tabular}
\end{table}

\paragraph{Observed qualitative pattern.}
\ref{tab:pipeline-sanity} shows the intended collapse-vs-noncollapse
behavior: representative pseudorandom and structured CNF families yield small
SPDP rank well below $\lceil\sqrt{n}\rceil$, while algebraically rigid
controls (a diagonal high-degree sum and the $3\times 3$ permanent) do
not collapse. This supports that the code path implements the same SPDP
objects defined in \ref{sec:spdp-matrix-rank-codim} and that the end-to-end
pipeline produces the expected pattern in a controlled setting.

\section{Conclusion}
We defined SPDP spaces and SPDP matrices and developed SPDP rank and codimension
as dual complexity invariants. We proved foundational properties (monotonicity,
symmetry invariance, and representation robustness) and stated generic circuit
upper-bound templates together with codimension reformulations. This establishes
SPDP rank/codimension as a coherent, standalone complexity-measure toolkit.


\end{document}